\documentclass{article}
\usepackage{amsthm,amsmath,amssymb,verbatim,algorithmic,algorithm,tikz,mathdots,mathpartir,stmaryrd,fancyvrb,MnSymbol,RR,RRthemes,hyperref}
\pdfoutput=1

\newtheorem{theorem}{Theorem}[section]
\newtheorem{lemma}[theorem]{Lemma}
\newtheorem{proposition}[theorem]{Proposition}
\newtheorem{corollary}[theorem]{Corollary}

\newcommand{\node}{n}  % nodes in a tree
\newcommand{\tree}{T}   % tree structures
\newcommand{\Nodes}{N}   % set of nodes of  a tree structure
\newcommand{\brel}{R}   % binary relation in a tree
\newcommand{\tlabel}{L}  % label function in a tre
\newcommand{\pathh}[3]{#1 \stackrel{#2}{\longrightarrow} #3}  % a path from #1 to #2 throught #3
\newcommand{\natn}{k}  % natural number
\newcommand{\Modals}{M} % the set of modalities
\newcommand{\Props}{P} % the set of propositions
\newcommand{\trail}{\alpha}  % a trail in a tree
\newcommand{\Trails}{Tl}  %  set of trails
\newcommand{\emptytrail}{\epsilon}
\newcommand{\dual}[1]{\overline{#1}}
\newcommand{\eqdef}{\stackrel{\text{\tiny{def}}}{=}}

\newcommand{\syntaxdef}{\mathrel{::=}}

\newcommand{\syntaxtable}[1]{
  \def\entry##1[##2]##3[##4]{
    {##1} & \syntaxdef& \hspace{3cm} & \!\!\!\! \mbox{##2}
    \\    &     & {##3} & \mbox{##4} }
  \def\singleentry##1[]##2[##3]{
  {##1} & \syntaxdef& {##2} & \!\!\!\! \mbox{##3} }
  \def\oris##1[##2]{
    \\    & |   & {##1} & \mbox{##2} }
  \def\orisopt##1[##2]{
    \\ \left(& |   & {##1} & \mbox{##2} \right) }
  \begin{array}{rcll}
  #1
  \end{array}
  }
\newcommand{\smallsyntax}[1]{\[\syntaxtable{#1}\]}
\newcommand{\modals}{m}  % modal symbol (modality)
\newcommand{\Formulas}{\Phi}  % the set of formulas

\newcommand{\true}{\top}  % the true symbol
\newcommand{\var}{x}  % variables
\newcommand{\prop}{p} % propositions
\newcommand{\form}{\phi}  % formulas
\newcommand{\restrictedform}{\psi}  % restricted formulas
\newcommand{\modalf}[2]{\langle #1\rangle #2} % modal formulas  (modality, formula)
\newcommand{\ufixpf}[2]{\mu #1.#2}  % unary fixpoint formula (variable, formula)
\newcommand{\countf}[5]{\langle {#2} \rangle^{#1}_{#4#5}#3}

\newcommand{\valuation}{V}  % valuation relation among nodes and variables
\newcommand{\semf}[3]{[\![#1]\!]^{#2}_{#3}}  % semantics of a formula #1 w.r.t tree #2 and a valuation #3
\newcommand{\semfc}[3]{\semf{#1}{#2_c}{#3}}

\newcommand{\flcl}[1]{{FL}({#1})}   % fisher-ladner closure
\newcommand{\lean}[1]{{Lean}({#1})}  % lean set
\newcommand{\fnode}[1]{n^{#1}}   % formula-node
\newcommand{\fNodes}[1]{N^{#1}}  % multiset of formula-nodes
\newcommand{\fbrel}[4]{R^{#1}(#2,#3)= #4}  % transition relation among formula-nodes through modalities (formula, node1, modality, node2)
\newcommand{\stree}{\Gamma}   % syntactic tree
\newcommand{\stroot}[1]{root(#1)}   %   the root of a syntactic tree
\newcommand{\stnodes}[1]{nodes(#1)}  % the nodes of a syntactic tree
\newcommand{\nav}[2]{nav((#1),#2)}  % (trail, formula)
\newcommand{\fishrel}[2]{ R^{fl}(#1,#2)}

\newcommand{\pathvar}{\rho} % path expression variable
\newcommand{\qvar}{q} % qualifier variable
\newcommand{\axvar}{a} % axis variable
\newcommand{\ch}{\text{child}}
\newcommand{\self}{\text{self}}
\newcommand{\pnt}{\text{parent}}
\newcommand{\desc}{\text{descendant}}
\newcommand{\descsf}{\text{desc$\!-\!$or$\!-\!$self}}
\newcommand{\anc}{\text{ancestor}}
\newcommand{\ancsf}{\text{anc$\!-\!$or$\!-\!$self}}
\newcommand{\fsib}{\text{foll$\!-\!$sibling}}
\newcommand{\psib}{\text{prec$\!-\!$sibling}}
\newcommand{\sibs}{\text{siblings}}
\newcommand{\foll}{\text{following}}
\newcommand{\prdn}{\text{preceding}}

\newcommand{\syntaxdefinition}{::=}

\newcommand{\XPath}{\text{XPath}}
\newcommand{\Axis}{\text{Axis}}
\newcommand{\PathExpr}{\text{PathExpr}}
\newcommand{\Step}{\text{Step}}
\newcommand{\NameTest}{\text{NameTest}}

\newcommand{\Qualifier}{\text{Qualifier}}
\newcommand{\QName}{\text{QName}}
\newcommand{\pathexcept}{~\text{except}~}
\newcommand{\pathunion}{~\text{union}~}
\newcommand{\pathintersect}{~\text{intersect}~}
\newcommand{\qualifnot}{\text{not}~}
\newcommand{\qualifand}{~\text{and}~}
\newcommand{\qualifor}{~\text{or}~}
\newcommand{\CountExpr}{\text{CountExpr}}
\newcommand{\Comparison}{\text{Comp}}
\newcommand{\qualifposition}{\text{position}()}
\newcommand{\qualifcount}[1]{\text{count}(#1)}

\newcommand{\precedence}{\ll}

\newcommand{\dom}{N}
\newcommand{\pathsem}[1]{\llbracket #1 \rrbracket}
\newcommand{\qualifsem}[1]{\llbracket #1 \rrbracket_{\text{Qualif}}}
\newcommand{\card}[1]{|#1|}

\newcommand{\fc}{\medtriangledown}%\downarrow}%\swarrow}%\text{first}}%
\newcommand{\ns}{\medtriangleright}%\text{next}}%
\newcommand{\invfc}{\medtriangleup}%\uparrow}\dual{\fc}}%
\newcommand{\invns}{\medtriangleleft}%\dual{\ns}}

\newcommand{\mucalcEFunc}{E^\rightarrow}
\newcommand{\mucalcPFunc}{P^\rightarrow}
\newcommand{\mucalcAFunc}{A^\rightarrow}
\newcommand{\mucalcE}[2]{\mucalcEFunc\llbracket{#1}\rrbracket_{#2}} 
\newcommand{\mucalcP}[2]{\mucalcPFunc\llbracket{#1}\rrbracket_{#2}}
\newcommand{\mucalcA}[2]{\mucalcAFunc\llbracket{#1}\rrbracket_{#2}}

\newcommand{\mucalcPexFunc}{P^\leftarrow}
\newcommand{\mucalcQexFunc}{Q^\leftarrow}
\newcommand{\mucalcAexFunc}{A^\leftarrow}
\newcommand{\mucalcPex}[2]{\mucalcPexFunc\llbracket{#1}\rrbracket_{#2}}
\newcommand{\mucalcQex}[2]{\mucalcQexFunc\llbracket{#1}\rrbracket_{#2}}
\newcommand{\mucalcAex}[2]{\mucalcAexFunc\llbracket{#1}\rrbracket_{#2}}

\newcommand{\nominal}{@n}
\newcommand{\et}{\wedge}
\newcommand{\ou}{\vee}

\newcommand{\someFCverifies}{\left<\fc\right>}
\newcommand{\someNSverifies}{\left<\ns\right>}

\newcommand{\someinvFCverifies}{\left<\invfc\right>}
\newcommand{\someinvNSverifies}{\left<\invns\right>}

\newcommand{\step}[2]{\text{{#1}::}{#2}}
\newcommand{\axis}[1]{\text{#1}}
\newcommand{\axisvar}{\emph{a}}

\newcommand{\qualif}[2]{{#1}\text{[$#2$]}}
\newcommand{\op}[1]{\mathbin{\text{\small{#1}}}}

\newcommand{\nodelabel}{\sigma}

\newcommand{\startatom}{\circledS}

\newcommand{\extracttrail}[1]{\text{trail}(#1)}
\newcommand{\target}[1]{\text{tail}(#1)}
\newcommand{\head}[1]{\text{head}(#1)}

\newcommand{\subst}[2]{^{#1} \!/\! _{#2}}
\newcommand{\psitree}[1]{$#1$tree}
\newcommand{\pathn}{\rho}
\newcommand{\nmax}{\mathtt{nmax}}
\newcommand{\imax}{\mathtt{max}}

\newcommand{\slf}[1]{\text{sf}(#1)}

\newcommand{\uc}{\text{ch}}

\newcommand{\induced}{\stackrel{.}{\in}}

\begin{document}
\RRNo{7251}

\RRtitle%{A Tree Logic with Graded Multidirectional Paths}
{A Tree Logic with Graded Paths and Nominals}%}%: Decidability and Succinctness} %Reasoning with 
%{Regular Path Expressions with Counting: a Succinct Tree Logic with Counting Along Multidirectional Paths}

\RRetitle%{A Tree Logic with Graded Multidirectional Paths}
{A Tree Logic with Graded Paths and Nominals}%}%: Decidability and Succinctness} %Reasoning with 
%{Regular Path Expressions with Counting: a Succinct Tree Logic with Counting Along Multidirectional Paths}

\RRauthor{Everardo B\'arcenas 
\and Pierre Genev\`es 
\and Nabil Laya\"ida 
\and Alan Schmitt 
}

\RRresume{Ce document introduit une logique d'arbre d\'ecidable en temps exponentielle et qui est capable d'exprimer des contraintes de cardinalit\'e sur chemins multidirectionnelle}
\RRkeyword{Modal Logic, XML, XPath, Schema}
\RRmotcle{Logique Modal, XML, XPath, Schema}
\RRprojets{WAM et SARDES}
%\RRdomaine{2} % cas du domaine numero 1
\RRthemeProj{wam} % theme du projet Apics
%\RRdomaineProjBis{wam}
\RRdate{April 2010}

\URRhoneAlpes

\RRabstract{
Regular tree grammars and regular path expressions constitute core constructs widely used in programming languages and type systems. %, and databases.
%Regular path expressions and tree grammars constitute the core constructs of XML programming languages (\`a la XDuce, CDuce,  XSLT, or XQuery). 
Nevertheless, there has been little research so far on reasoning frameworks for path expressions where node cardinality constraints occur along a path in a tree. %, with the notable exceptions of \cite{964013,DBLP:conf/icalp/SeidlSMH04}.
%These works however consider a form of counting limited to sibling nodes.
 We present a logic capable of expressing deep counting along paths which may include arbitrary recursive forward and backward navigation.
The counting extensions can be seen as a generalization of graded modalities that count immediate successor nodes. While the combination of graded modalities, nominals, and inverse modalities yields undecidable logics over graphs, we show that these features can be combined in a tree logic decidable in exponential time. %We use it for the static analysis of XPath.
%Specifically, we developed a tableau-based satisfiability-checking algorithm for the logic that has the same optimal computational complexity than the logic without counting operators, that is, a simple exponential w.r.t the formula size.
}

\makeRR

\section{Introduction}

%- expressivity: starting from a single node described by a regular path, we can count along another arbitrary regular path.......
%(other works not really related to xpath..)
%
%- succinctness
%
%- decidability over trees whereas undecidable over graphs

%This work is motivated by the need for incorporating XML as first-class constructs in programming languages (\cite{767195,DBLP:conf/icfp/BenzakenCF03}). Our goal is to build a new generation of tools capable of considering  reliability (type-safety) and performance (optimization) issues at compile time. 

A fundamental peculiarity of XML is the description of regular properties. For example, in XML schema languages the content types of element definitions is made through the use of regular expressions. In addition, selecting nodes in such constrained trees is also done by the mean of regular path expressions (\`a la XPath). In both cases, it is often interesting to be able to express conditions on the frequency of occurrences of nodes. 
%We present in this work a succinct representation of such frequency conditions by means of a modal logic for trees.

Even if we consider simple strings, it is well known that some formal languages easily described in English may require voluminous regular expressions. For instance, as pointed in \cite{klarlund-tacas95}, the language $L_{2a2b}$ of all strings over $\Sigma=\{a,b,c\}$ containing at least two occurrences of $a$ \emph{and} at least two occurrences of $b$ seems to require a large expression, such as:
\begin{align*}
&& \Sigma^*a\Sigma^*a\Sigma^*b\Sigma^*b\Sigma^* &&\cup&& \Sigma^*a\Sigma^*b\Sigma^*a\Sigma^*b\Sigma^* \\
%-----
&\cup& \Sigma^*a\Sigma^*b\Sigma^*b\Sigma^*a\Sigma^* 
&&\cup&& \Sigma^*b\Sigma^*b\Sigma^*a\Sigma^*a\Sigma^* \\
%-----
&\cup& \Sigma^*b\Sigma^*a\Sigma^*b\Sigma^*a\Sigma^* 
&&\cup&& \Sigma^*b\Sigma^*a\Sigma^*a\Sigma^*b\Sigma^*.
\end{align*}
If we added $\cap$ to the operators for forming regular expressions, then the language $\{a,b,c\}$ could be expressed more concisely as $(\Sigma^*a\Sigma^*a\Sigma^*) \cap (\Sigma^*b\Sigma^*b\Sigma^*)$. In logical terms, conjunction offers a first dramatic reduction in expression size.

%Even with this enriched set of operators, it is often more convenient to express regular languages in terms of positions and corresponding symbols, for which  logics often happen to be especially appropriate, as pointed out in \cite{klarlund-tacas95}.
%Logics are often more convenient to express regular languages.

If we now consider a formalism equipped with the ability of describing numerical constraints on the frequency of occurrences, we get a second (exponential) reduction in size. For instance, the above expression can be formulated as
 $(\Sigma^*a\Sigma^*)^2 \cap (\Sigma^*b\Sigma^*)^2$.
We can even write  $(\Sigma^*a\Sigma^*)^{2^{20}} \cap (\Sigma^*b\Sigma^*)^{2^{20}}$ instead of a (much) larger expression.

%$\countf{}{\ns^*}{a}{>}{2} \wedge \countf{}{\ns^*}{b}{>}{2} \wedge \neg \modalf{(\ns)^*}{\neg(a\vee b\vee c)} $
%where next makes explicit the successor relation over strings.

Different extensions of regular expressions with intersection, counting constraints, and interleaving have been recently considered over strings, and for describing content models of sibling nodes in XML type languages \cite{ghelli-icdt09,Gelade-siam08,Kilpelainen-ic07}.  The complexity of the inclusion problem over these different language extensions and their combinations typically ranges from polynomial to exponential space (see \cite{Gelade-siam08} for a survey). The main distinction between these works and the  work presented here is that we focus on counting nodes located along deep and recursive paths in trees.

When considering regular \emph{tree} languages instead of regular \emph{string} languages,  succinct syntactic sugars such as the ones presented above are even more useful, as branching makes the situation more combinatorial compared to strings. 
In the case of trees, it is often useful to express cardinality constraints not only on the sequence of children nodes, but also in a particular region of a tree: in a subtree for example. Suppose for instance that we want to define a tree language over $\Sigma$ where there is no more than 2 ``b'' nodes. This seems to require a quite large regular tree type expression such as the one below:
\newcommand{\Xstart}{x_\text{root}}
\newcommand{\Xtwobmax}{x_{b\leq2}}
\newcommand{\Xnob}{x_{\neg b}} %\text{\sout{\ensuremath{b}}}}}
\newcommand{\Xonebmax}{x_{b\leq1}}
\newcommand{\lb}{\texttt{[}}
\newcommand{\rb}{\texttt{]}}
\newcommand{\xtag}[2]{#1 \lb #2 \rb}
\newcommand{\Xsimple}{x}
$$\begin{array}{lcl}
\Xstart  \!\!\! &\rightarrow&\!\!\ \xtag{b}{\Xonebmax} \mid \xtag{c}{\Xtwobmax} \mid \xtag{a}{\Xtwobmax} \vspace{0.1cm} \\
\Xtwobmax\!\!\! &\rightarrow&\!\!\! \Xnob, \xtag{b}{\Xnob},\Xnob,\xtag{b}{\Xnob},\Xnob \mid \Xnob, \xtag{b}{\Xonebmax},\Xnob  \vspace{0.1cm}\\
&& \mid \Xnob, \xtag{a}{\Xtwobmax}, \Xnob \mid \Xnob, \xtag{c}{\Xtwobmax}, \Xnob  \mid \Xonebmax \vspace{0.1cm}\\
\Xonebmax \!\!\! &\rightarrow&\!\!\ \Xnob \mid  \Xnob, \xtag{b}{\Xnob}, \Xnob \mid \xtag{a}{\Xonebmax} \mid \xtag{c}{\Xonebmax}  \vspace{0.1cm}\\
\Xnob \!\!\! &\rightarrow&\!\!\ (\xtag{a}{\Xnob} \mid \xtag{c}{\Xnob})^*
\end{array}$$
where $\Xstart$ is the starting non-terminal; $\Xnob, \Xonebmax, \Xtwobmax$ are non-terminals; and the bracket notation $\xtag{a}{\Xnob}$ describes a subtree whose root is labeled $a$ and in which there is no $b$ node.

More generally, the widely adopted notations for regular tree grammars produce very verbose definitions for properties involving cardinality constraints on the nesting of elements\footnote{This is typically the reason why the standard DTD for XHTML does not syntactically prevent the nesting of anchors, whereas this nesting is actually prohibited in the XHTML standard.}.

The problem with regular tree (and even string) grammars is that one is forced to fully expand all the patterns of interest using concatenation, union, and Kleene star. Instead, it is often tempting to rely on another kind of (formal) notation that just describes a simple pattern and additional constraints on it. For instance, one could imagine denoting the previous example as follows, where the additional constraint is described using XPath notation:
$$\left(\Xsimple  \!\rightarrow\!\! (\xtag{a}{\Xsimple} \mid \xtag{b}{\Xsimple} \mid \xtag{c}{\Xsimple})^*\right) \vspace{0.1cm}  ~\et~ \text{count(/descendant-or-self::}b) \leq 2$$

Although this kind of counting operators does not increase the expressive power of the regular tree grammars, they can have a drastic impact on succinctness, thus making reasoning over these languages harder (as noticed in \cite{DBLP:conf/mfcs/Gelade08} in the case of strings).
%This approach has been explored for the case of string regular expressions (see \cite{DBLP:conf/mfcs/GeladeGM09} and references thereof). 
%
Indeed, reasoning on this kind of extensions without relying on their expansion (in order to avoid syntactic blow-ups) is often tricky \cite{DBLP:conf/mfcs/GeladeGM09}. Determining satisfiability, containment, and equivalence over these classes of extended regular expressions typically require involved algorithms with extra-complexity \cite{DBLP:conf/focs/MeyerS72} compared to plain vanilla regular expressions.

%In the present paper, we propose a logical notation that happens to be especially appropriate for describing many sorts of cardinality constraints on the frequency of occurrence of nodes in trees.
In the present paper, we propose a logical notation that happens to be especially appropriate for describing many sorts of cardinality constraints on the frequency of occurrence of nodes in 
regular tree types. Regular tree types encompass most of XML  types  (DTDs, XML Schemas, RelaxNGs) used in practice today.

XPath is the standard query language for XML documents, and it is an important part of other XML technologies such as XSLT and XQuery.
XPath expressions are regular path expressions interpreted as sets of nodes selected from a given context node. 
In contrast with regular tree types, which only express properties on children nodes, most of the expressive power of XPath comes from the ability to perform multidirectional navigation,
that is, XPath expressions are able to express properties involving not only recursive navigation, as for descendant nodes for instance, but also backward navigation, as for ancestor nodes. 
Unfortunately, expressing cardinality restrictions on nodes accessible by recursive multidirectional paths may introduce an extra-exponential cost \cite{DBLP:conf/doceng/GenevesR05,Balder09}, 
or may even lead to undecidable formalisms \cite{Balder09,DBLP:conf/cade/DemriL06}.
We propose in this paper a decidable framework capable of succinctly express cardinality constraints along deep multidirectional paths.

\paragraph{Contribution and Outline}
We introduce a tree logic with counting operators for expressing arbitrarily deep and recursive counting constraints in Section~\ref{sec:logic}. A sound and complete algorithm for testing satisfiability of logical formulas in exponential time is presented in Section~\ref{sec:algo}. Section~\ref{sec:application} shows how the logic and the algorithm can be applied in the XML setting and in particular for the static analysis of XPath expressions and common schemas containing constraints on the frequency of occurrence of nodes. Finally, we review related works in Section~\ref{sec:relatedwork} before concluding in Section~\ref{sec:conclusion}.

\section{Counting Tree Logic} \label{sec:logic}
We first present trees that we consider, and define a notion of trails in trees, before introducing the syntax and semantics of logical formulas.

%%%%%%%%%%%%%%%%%%%%%%%%%%%%%%%%%%%%%%%%%%%%%%%%%%%%%%%%%%%%
\subsection{Trees}\label{subsec:trees}
 %Trails among two nodes in a tree are defined in order to introduce a counting operator able to perform a deep-like counting.
%Models of the logic are trees.\as{Do we need to say this now? We repeat that later, when defining the interpretation of formulas.} 
%
We consider finite trees which are node-labeled and sibling-ordered. Since there is a well-known bijective encoding between $n\!-\!$ary and binary trees, we focus on binary trees without loss of generality.
Specifically, we use the encoding represented in Figure \ref{fig:depthlevels}, 
where the binary representation preserves the first child of a node and append sibling nodes as second successors.
\begin{figure}
\begin{center}
\includegraphics[keepaspectratio,width=7cm]{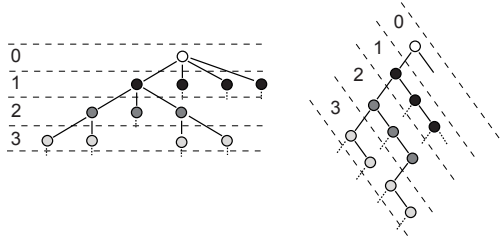}
\end{center}
\caption{$n\!-\!$ary to binary trees}
\label{fig:depthlevels}
\end{figure}

We consider the modalities ``$\fc$'' and  ``$\ns$''.  The modality  ``$\fc$'' labels the edge between a node and its first child. The modality  ``$\ns$'' labels the edge between a node and its next sibling.
We also consider the converse modalities  ``$\invfc$'' and  ``$\invns$'' that respectively labels the same edges in the reverse direction. 

In order to define a simple set theoretic semantics for the logic, we consider trees in a way similar to Kripke structures for modal logics \cite{685998}. Specifically, we name  $\Modals=\{\fc,\ns,\invfc,\invns\}$ the set of {\em modalities}. For $m\in \Modals$ we denote by $\dual{m}$ the corresponding inverse modality ($\dual{\fc}=\invfc, \dual{\ns}=\invns, \dual{\invfc}=\fc, \dual{\invns}=\ns$). We consider a countable alphabet $\Props$ of {\em propositions} representing names of nodes. A node is labeled with exactly one proposition.

A tree can then be seen as a tuple $(\Nodes,\brel,\tlabel)$, where:
$\Nodes$ is a finite set of nodes; $\brel$ is a partial mapping from $\Nodes\times\Modals$ to $\Nodes$ that restricts the labeling of edges to form a tree structure; and $\tlabel$ is a labeling function from $\Nodes$ to $\Props$.

\subsection{Trails}\label{subsec:trails}

Trails are defined as regular expressions formed by modalities, as follows: %in Figure \ref{trailsyn}.
%\begin{comment} %%%%%%%%%%%%%%%%%%%%%%%%%%%%%%%%%%%%%%%%%%%%%%%%%%%%%%%%%
%%\begin{figure}
%\begin{align*}
%& \Trails \ni \trail &&:=&& &&\text{trail}\\
%& && && \emptytrail &&\text{empty}\\
%& &&|&& \modals     &&\text{modality}\\
%& &&|&& \trail_1,\trail_2 && \text{concatenation}\\
%& &&|&& \trail_1|\trail_2 && \text{disjunction}\\
%& &&|&& \trail^\star      && \text{recursion}
%\end{align*}
%\end{comment}%%%%%%%%%%%%%%%%%%%%%%%%%%%%%%%%%%%%%%%%%%%%%%%%%%%%%%%%%%
\begin{align*}
\trail_0 &::= \modals \mid \trail_0,\trail_0 \mid \trail_0\shortmid\trail_0\\
\trail &::= \trail_0 \mid \trail_0^\star \trail_0
\end{align*}
%\caption{Syntax of Trails.} \label{trailsyn}
%\end{figure}
%Two interpretations of trails are considered: a syntactic one and a semantic one.
We restrict trails to sequences or repeated subtrails (which contain no repetition) followed by a subtrail (with no repetition). We also disallow trails of the form $\modals, \dual\modals$, which may result in formulas with cycles.

The syntactic interpretation of trails corresponds to sets of sequences of modalities (as in the usual semantics of regular expressions).
% whereas the semantic one corresponds to pairs of nodes joined by a trail in a tree.

In a given tree, we say that there is a {\em trail $\trail$ from the node $\node_0$ to the node $\node_\natn$}, written $\pathh{n_0}{\trail}{n_\natn}$,
if and only if there is a sequence of nodes $\node_0, \ldots, n_k$ and a sequence of modalities $\modals_1, \ldots, \modals_k$ that belongs to the syntactic interpretation of the trail $\trail$, such that $\brel(\node_j,\modals_{j+1})=\node_{j+1}$, where $j=0,\ldots,\natn-1$.  
We say that a path $\pathn$ among two nodes belongs to a trail $\trail$, written $\pathn \in \trail$, if there exists a sequence of modalities between the nodes that belongs to the interpretation of the trail.

%Given a tree $\tree$ and a node $\node$ in $\tree$, we interpret  a trail $\trail$ as follow: $\semf{\trail}{\tree}{\node}=\{\node^\prime\mid \path{\node}{\trail}{\node^\prime}\}$.

\subsection{Syntax of Logical Formulas}\label{subsec:syntax}

The syntax of logical formulas is given in Figure \ref{normalformsyn}, where $m \in \Modals$ and $k \in \mathbb{N}$. The syntax is shown in negation normal form, which can be reached usual De Morgan rules together with rules given in Figure~\ref{normalneg}. The fact that the semantic interpretation is preserved even though the smallest fixpoint does not become a greatest fixpoint is a consequence of Lemma~\ref{lem:finite_unfolding}.
\begin{figure}
  
  \smallsyntax{
  \Formulas \ni \entry       \form    [formula]
    \top  \quad | \quad  \neg\top           [true, false]
  \oris \prop ~\quad | \quad \neg \prop             [atomic prop (negated)]
  \oris        x [recursion variable]
  \oris        \phi \ou \phi [disjunction]
  \oris        \phi \et \phi [conjunction]
  \oris  \modalf{\modals}\phi  \quad | \quad  \neg \modalf{\modals}\true [modality (negated)]

  \oris \countf{}{\trail}{\restrictedform}{\leq}{\natn}  \quad | \quad \countf{}{\trail}{\restrictedform}{>}{\natn} [counting]

  \oris        \ufixpf{\var}{\restrictedform} [fixpoint operator] \\
  \entry       \restrictedform    []
    \top  \quad | \quad  \neg\top           []
  \oris \prop ~\quad | \quad \neg \prop             []
  \oris        x []
  \oris        \restrictedform \ou \restrictedform []
  \oris        \restrictedform \et \restrictedform []
  \oris  \modalf{\modals}\restrictedform  \quad | \quad  \neg \modalf{\modals}\true []
  \oris        \ufixpf{\var}{\restrictedform} []
  } 
\caption{Syntax of Formulas (in Normal Form).} \label{normalformsyn}
\end{figure}

\begin{figure}
\begin{align*}
\neg \modalf{\modals}{\form}  &\equiv  \neg \modalf{\modals}{\true} \vee \modalf{\modals}{\neg \form} 
& \neg \ufixpf{\var}{\restrictedform}  &\equiv \ufixpf{\var}{\neg \restrictedform\{\subst{\var}{\neg \var}\}}\\
 \neg \countf{}{\trail}{\restrictedform}{\leq}{\natn} &\equiv \countf{}{\trail}{\restrictedform}{>}{\natn} &  \neg \countf{}{\trail}{\restrictedform}{>}{\natn} &\equiv \countf{}{\trail}{\restrictedform}{\leq}{\natn}
\end{align*}
\caption{Reduction to Negation Normal Form.}\label{normalneg}
\end{figure}

%In the following, we write $\countf{}{\trail}{\form}{>}{\natn}$ for $\neg \countf{}{\trail}{\form}{\leq}{\natn}$, and $\modalf{\modals}{\form}$ for $\countf{}{\modals}{\form}{>}{0}$.

Defining an {\em equality} operator for counting formulas is straightforward.
\begin{align*}
\countf{}{\trail}{\restrictedform}{=}{\natn} &\equiv \countf{}{\trail}{\restrictedform}{>}{(\natn-1)} \wedge \countf{}{\trail}{\restrictedform}{\leq}{\natn} &&\text{if $k>0$}\\
\countf{}{\trail}{\restrictedform}{=}{0} &\equiv  \countf{}{\trail}{\restrictedform}{\leq}{0}
\end{align*}

% In the following, we will also write $\false$ instead of $\neg\true$.

\subsection{Semantics of Logical Formulas}\label{subsec:semantics}
Formulas are interpreted as sets of nodes in a tree. A model of a formula is a tree, such that the formula denotes
a non-empty set of nodes in this tree. A counting formula $\countf{}{\trail}{\restrictedform}{>}{\natn}$ is interpreted as follows:
the set of nodes such that there are at least $\natn + 1$  nodes satisfying $\restrictedform $ through the trail $\trail$.
For example, the formula $\prop_1\wedge \modalf{\fc}{\countf{}{\ns^*}{\prop_2}{>}{5}}$, denotes $\prop_1$ nodes
with strictly more than $5$ children nodes named $\prop_2$.

In order to present the formal semantics of formulas, we introduce valuations. % and extension of valuations.
Given a tree, a {\em valuation } $\valuation$ is a binary relation between tree nodes and variables.
We write $\valuation[\subst{\Nodes^\prime}{x}]$, where $\Nodes^\prime$ is a subset of the nodes, for the relation denoted by $V$ extended with $(n,x)$ for every  $\node\in\Nodes^\prime$.
%
%NDP: clean the definition of valuations: suggestion: A valuation $\valuation$ is an assignment of subsets of $\Nodes$ to variables....  
%
Given a tree $\tree=(\Nodes,\brel, \tlabel)$ and a valuation $\valuation$, the formal semantics of formulas is given in Figure~\ref{formsem}.
\begin{figure}
\begin{align*}
& \semf{\true}{\tree}{\valuation} &&=&& \Nodes \\
& \semf{\neg \true}{\tree}{\valuation} &&=&& \emptyset \\
& \semf{\prop}{\tree}{\valuation} &&=&& \{\node , \tlabel(\node)=\prop\} \\
& \semf{\neg \prop}{\tree}{\valuation} &&=&& \{\node , \tlabel(\node)\neq \prop\} \\
& \semf{\var}{\tree}{\valuation} &&=&& \{\node , (\node,\var)\in \valuation\}\\
& \semf{\form_1\vee \form_2}{\tree}{\valuation} &&=&& \semf{\form_1}{\tree}{\valuation} \cup \semf{\form_2}{\tree}{\valuation}\\
& \semf{\form_1\wedge \form_2}{\tree}{\valuation} &&=&& \semf{\form_1}{\tree}{\valuation} \cap \semf{\form_2}{\tree}{\valuation}\\
& \semf{\modalf{\modals}{\form}}{\tree}{\valuation} &&=&& \{\node , \brel(\node,\modals) \in \semf{\form}{\tree}{\valuation}\}\\
& \semf{\neg \modalf{\modals}{\true}}{\tree}{\valuation} &&=&& \{\node , \brel(\node,\modals) \text{ undefined}\} \\
& \semf{\countf{}{\alpha}{\restrictedform}{\leq}{\natn}}{\tree}{\valuation} &&=&&
    \{\node , |\{\node^\prime\in \semf{\restrictedform}{\tree}{\valuation} \mid \pathh{\node}{\alpha}{\node^\prime}\}|\leq  \natn \}\\
& \semf{\countf{}{\alpha}{\restrictedform}{>}{\natn}}{\tree}{\valuation} &&=&&
    \{\node ,  |\{\node^\prime\in \semf{\restrictedform}{\tree}{\valuation} \mid \pathh{\node}{\alpha}{\node^\prime}\}|>  \natn \}\\
& \semf{\ufixpf{\var}{\restrictedform}}{\tree}{\valuation} &&=&& \bigcap\{\Nodes^\prime , \semf{\restrictedform}{\tree}{\valuation[\subst{\Nodes^\prime}{x}]}\subseteq \Nodes^\prime\}
\end{align*}
\caption{Semantics of Formulas.} \label{formsem}
\end{figure}

Intuitively, the formulas are interpreted as sets of nodes in a tree: propositions denote the nodes where they occur;
negation is interpreted as set complement; disjunction and conjunction are respectively set union and intersection;
the least fixpoint operator performs finite recursive navigation; 
and the counting operator denotes certain nodes, named the source nodes, such that the nodes, accessible from a single source through a trail, fulfill a cardinality restriction.  A formula is said to be {\em satisfiable} when its interpretation is not empty. 
%there exists a tree for which it does not denote the empty set.

%When two formulas denote the same set of nodes for every tree, we say such formulas are {\em equivalent}.\as{I guess equivalent should be ``for every tree''.}

\subsection{Restriction over Formulas} \label{cyclefreeness}
We consider a syntactic restriction over formulas similar to the one in \cite{geneves-pldi07}: every formula of the logic must be \emph{cycle-free} (so that the logic is closed under negation \cite{geneves-pldi07}). Intuitively, in a cycle-free formula, fixpoint variables do not occur in the scope of both a modality and its converse.  For example, cycle-free trails are trails where both a subtrail and its converse do not occur under the scope of the recursion operator.  We do not consider counting formulas under fixpoints nor under counting formulas.
%Under recursion we consider only counting formulas whose trail

\begin{lemma}\label{lem:finite_unfolding}
  Let $\form$ be a cycle-free formula, and $\tree$ be a tree for which $\semf{\form}{\tree}{\emptyset} \neq \emptyset$. Then there is a finite unfolding $\form^\prime$ of the fixpoints of $\form$ such that $\semf{\form^{\prime}\{\subst{\neg\true}{\ufixpf{\var}{\restrictedform}}\}}{\tree}{\emptyset} = \semf{\form}{\tree}{\emptyset}$.
\end{lemma}

\begin{proof}
  As counting formulas may be replaced by non-counting formulas (with the cost of an exponential blow up), the proof is identical to the one in \cite{geneves-pldi07}.
\end{proof}

\subsection{Global Counting Formulas and Nominals}\label{subsec:globalformula}

An interesting consequence of the inclusion of backward axes in trails is the ability to reach every node in the tree from a given node of the tree, using the trail $(\invfc|\invns)^\star,(\fc|\ns)^\star$\footnote{Note that this trail is cycle-free.}. We can thus select some nodes depending on some global counting property. Consider the following formula, where $\#$ stands for one of the comparison operators $\leq,>,=$.
\[\countf{}{(\invfc|\invns)^\star,(\fc|\ns)^\star}{\form_1}{\#}{\natn}\]
Intuitively, this formula considers each node $n$ of the tree, and counts how many nodes in the whole tree satisfy $\form_1$. It then selects node $n$ if and only if the count is compatible with the comparison considered. This formula thus returns either every node of the tree, or the empty set.
It is then easy to restrict the selected nodes to some that satisfy a given formula $\form_2$, using intersection.
\[\countf{}{(\invfc|\invns)^\star,(\fc|\ns)^\star}{\form_1}{\#}{\natn} \wedge \form_2\]
This formula select every node satisfying $\form_2$ if and only if there are $\#\natn$ nodes satisfying $\form_1$, which we write as follows.
\[ \form_1\#\natn \implies \form_2 \]
We can now express existential properties, such as ``select all nodes satisfying $\form_2$ if there exists a node satisfying $\form_1$''.
\[ \form_1 > 0 \implies \form_2 \]
We can also express universal properties, such as ``select all nodes satisfying $\form_2$ if every node satisfies $\form_1$''.
\[ (\neg\form_1) \leq 0 \implies \form_2 \]

Another way to interpret global counting formulas is as a generalization of the so-called nominals in the modal logics community \cite{DBLP:conf/cade/SattlerV01}. Nominals are special propositions whose interpretation is a singleton (they occur exactly once in the model). They come for free with the logic. A nominal, denoted ``$\nominal$'' in the remaining part of the paper, corresponds to the following global counting formula:
\[\countf{}{(\invfc|\invns)^\star,(\fc|\ns)^\star}{n}{=}{1}\] where $n$ is a new fresh atomic proposition.

%We can also express a navigation to everywhere in the tree by the following fixpoint formula.
%\[EW(\form)= \ufixpf{\var_1}{(\ufixpf{\var_2}{\form \vee \modalf{\fc}{\var_2} \vee \modalf{\ns}{\var_2}}) \vee \modalf{\invfc}{\var_1} \vee \modalf{\invns}{\var_2}}\]
Notice that we can also perform a navigation to everywhere in a tree with only fixpoint formulas,
hence a nominal can be alternatively written as:
\begin{align*}
\nominal \equiv n \wedge\neg [& \desc(n) \vee \anc(n) \vee \\
                            &  \descsf(\sibs(n)) \vee  \\
							& \descsf(\sibs(\anc(n))) ], 
\end{align*}
where:
\begin{align*}
&\desc(\form)&&=&& \modalf{\fc}{\ufixpf{\var}{\form \vee \modalf{\fc}{\var} \vee \modalf{\ns}{\var}}}\\
& \fsib(\form) &&=&& \ufixpf{\var}{\modalf{\ns}{\form} \vee \modalf{\ns}{\var}}\\
& \psib(\form) &&=&& \ufixpf{\var}{\modalf{\invns}{\form} \vee \modalf{\invns}{\var}}\\
& \descsf(\form) &&=&& \ufixpf{\var_0}{\form \vee \modalf{\fc}{\ufixpf{\var_1}{\var_0 \vee \modalf{\ns}{\var_1}} }}\\
& \anc(\form) &&=&& \ufixpf{\var}{\modalf{\invfc}{(\form \vee \var)} \vee \modalf{\invns}{\var}}\\
& \sibs(\form) &&=&& \fsib(\form) \vee \psib(\form)
\end{align*}

\subsection{Graded Paths}\label{sec:gradedpaths}

%This can be seen as generalized graded modalities in trees. 
%. 

Graded modalities have been introduced to count immediate successor nodes in graphs \cite{DBLP:conf/cade/KupfermanSV02}. Specifically, graded modalities make it possible to restrict the number of occurrences of immediate successors of a node in a graph by the mean of an explicit constant upper-bound and/or lower-bound. Here we consider trees and extend the ``immediate successor'' notion to nodes reachable from any regular path, including reverse and recursive navigation. 

A peculiarity of graded modalities in graphs is that they can be used inside recursive formulas. A similar notion in trees consists in counting immediate children nodes, as performed by the counting formula $\modalf{\fc}{\countf{}{\ns^*}{\form}{\#}{k}}$, where $\phi$ describes the property to be counted. It is then possible to consider occurrences of this counting formula inside a fixpoint operator. This is because this peculiar counting formula can be simply rewritten in terms of plain vanilla logical formulas. For instance, the formula $\modalf{\fc}{\countf{}{\ns^*}{\prop}{>}{1}}$ states the existence of at least two ``$p$'' children, and is translated into:
\[\modalf{\fc}{\ufixpf{x}{(\prop \wedge \modalf{\fc}{\ufixpf{y}{\prop \vee \modalf{\ns}{y} }}) \vee \modalf{\ns}{x} }} 
\]
The general nesting scheme of this translation can be expressed as follows, where the function $\uc(\cdot)$ takes such a counting formula as input and returns its translation:
\begin{align*}
 \uc(\modalf{\fc}{\countf{}{\ns^*}{\form}{>}{0}}) =& \modalf{\fc}{\ufixpf{\var}{\form \vee \modalf{\ns}{\var}}}\\
 \uc(\modalf{\fc}{\countf{}{\ns^*}{\form}{>}{k+1}}) =& \modalf{\fc}{\ufixpf{\var}{ \left(\form \wedge\uc(\modalf{\fc}{\countf{}{\ns^*}{\form}{>}{k}})\right) \vee \modalf{\ns}{\var}}} \\
 \uc(\modalf{\fc}{\countf{}{\ns^*}{\form}{\leq}{k}})  =& \neg \uc(\modalf{\fc}{\countf{}{\ns^*}{\form}{>}{k}}) 
\end{align*}
We can even apply a recursive version of this transformation in order to rewrite nested counting formulas. 

In Lemma \ref{countinfixp}, we show the computational cost of the translation does not depend on the size of the formula, but on the nesting level of counting subformulas.

The possibility of using an arbitrary fixpoint operator around a given formula allows one to express the ``until'' operator, proposed for XPath by Marx \cite{marx-tods05}. 
Owing to the previous translation, we can combine counting features with the ``until'' operator and express properties that go beyond the expressive power of the XPath 1.0 standard. For instance, the following formula states that ``starting from the current node, until we reach an ancestor named $a$, every ancestor has at least 3 children named $b$'':
$$\ufixpf{x}{\left(\modalf{\fc}{\countf{}{\ns^*}{b}{>}{2}} \wedge \ufixpf{y}{\modalf{\invfc}{x}\vee \modalf{\invns}{y}}\right)  \vee a}$$

\section{Satisfiability Algorithm} \label{sec:algo}
We present a tableau-based algorithm for checking satisfiability of formulas. Given a formula, the algorithm seeks to build a satisfying tree. A satisfying tree is found if and only if the formula is satisfiable, otherwise the algorithm concludes that the formula is unsatisfiable.

\subsection{Overview}
The algorithm operates in two stages. 

First, a formula $\form$ is decomposed into a set of subformulas, called the \emph{Lean}. The Lean gathers all subformulas that are useful for determining the truth status of the initial formula, while eliminating redundancies. For instance, conjunctions and disjunctions are eliminated at this stage, since, if a subformula $\phi_1$ holds then one does not need to know the truth status of $\phi_2$ in order to determine the truth status of $\phi_1 \vee \phi_2$. In fact, the lean (defined in \ref{sec:lean}) only gathers atomic propositions and modal subformulas. The Lean defines a finite number of formulas that can be composed. The set of all these compositions represents the exhaustive search universe in which the algorithm is looking for a satisfying tree. A tree node corresponds to a valuation of the Lean formulas.

The second stage of the algorithm consists in a least fixpoint computation that builds every relevant binary tree in a bottom-up manner. At the first step of this stage, all possible leaves are considered. At each further step, the algorithm considers every possible parent node that can be connected with a node of the previous steps.
At each step, built subtrees are checked for consistency: for instance if a formula at a node $n$ involve a forward modality  $\modalf{\fc}{\form'}$, then $\form'$ must be verified at the first child of $n$. Reciprocally, due to converse modalities, a given node may impose restrictions on its possible parent nodes.
The algorithm only considers consistent nodes at each step, meaning that the whole subtree of a given node added at a given step provably satisfies a subformula, except its potential top-level backward modalities that will be taken into account at the next step. At each step, counting formulas are verified. 
Finally, the algorithm terminates whenever:
\begin{itemize}
  \item  either a tree that satisfies the initial formula has been found, and its root does not contain any pending (unproven) backward modality; or
  \item no more parent nodes can be considered (the exploration of the whole search universe is complete): the formula is unsatisfiable.
\end{itemize}

The algorithm is proven sound and complete: $\form$ is satisfiable if and only if a tree in which $\form$ is satisfied at some node is built. Thus either such a tree is built, or $\form$ is not satisfiable. 

%This approach is very similar to our previous work \cite{geneves-pldi07}, the main challenge is to correctly deal with counting formulas.

\subsection{Preliminaries}
We first annotate every counting formula with a fresh \emph{counting proposition} $c$, written $\countf{c}{\trail}{\form}{\#}{\natn}$.
We first formally define the notions of Lean and nodes. To this end, we first need to extract navigating formulas from counting formulas.
\begin{align*}
 nav(\var) &= \var \\
 nav(\prop) &= \prop \\
 nav(\top) &= \top \\
 nav(c) &= c\\
 nav(\neg \prop) &= \neg \prop \\
 nav(\neg \modalf{\modals}{\top})&= \neg \modalf{\modals}{\top}\\
 nav(\form_1\wedge \form_2)&= nav(\form_1) \wedge nav(\form_2) \\
 nav(\form_1\vee \form_2)&= nav(\form_1) \vee nav(\form_2) \\
 nav(\modalf{\modals}{\form}) &= \modalf{\modals}{nav(\form)} \\
 nav(\ufixpf{\var}{\restrictedform}) &= \ufixpf{\var}{nav(\restrictedform)}\\
 nav(\countf{c}{\trail}{\restrictedform}{>}{\natn})&= \nav{\trail}{\restrictedform\wedge c}\\
 nav(\countf{c}{\trail}{\restrictedform}{\leq}{\natn})&= \nav{\trail}{(\restrictedform \land c) \lor (\neg\restrictedform \land \neg c)}\\
%nav(\neg\countf{}{\trail}{\form}{\leq}{\natn}) &=\nav{\trail}{\form}\\
\nav{\epsilon}{\restrictedform} &= \restrictedform \\
 \nav{\modals}{\restrictedform} &= \modalf{\modals}{\restrictedform}\\
 \nav{\trail_1,\trail_2}{\restrictedform}&= \nav{\trail_1}{\nav{\trail_2}{\restrictedform}} \\
 \nav{\trail_1\mid\trail_2}{\restrictedform} &= \nav{\trail_1}{\restrictedform} \vee \nav{\trail_2}{\restrictedform}\\
 \nav{\trail^\star}{\restrictedform} &= \ufixpf{\var}{nav(\restrictedform)\vee\nav{\trail}{\var}}
\end{align*}

We define the {\em Fisher-Ladner} relation among formulas as follow, where $i=1,2$.
\begin{align*}
& \fishrel{\form_1\wedge\form_2}{\form_i}, && \fishrel{\form_1\vee\form_2}{\form_i}, \\
& \fishrel{\ufixpf{\var}{\form}}{\form[\subst{\ufixpf{\var}{\form}}{\var}]}, && \fishrel{\countf{c}{\trail}{\restrictedform}{\#}{\natn}}{nav(\countf{c}{\trail}{\restrictedform}{\#}{\natn})},\\
&  \fishrel{\modalf{\modals}\phi}{\phi}.
\end{align*}

The {\em Fisher-Ladner} closure of a formula $\form$, written $\flcl{\form}$, is the set defined as follow.
\begin{align*}
& \flcl{\form}_0 &&=&& \{\form\}, \\
& \flcl{\form}_{i+1} &&=&& \flcl{\form}_i \cup \{\form^\prime \mid \fishrel{\form^{\prime\prime}}{\form^\prime}, \form^{\prime\prime}\in \flcl{\form}_i\},\\
& \flcl{\form}&&=&&\flcl{\form}_k,
\end{align*}
\text{where $k$ is the smallest integer s.t. $\flcl{\form}_k=\flcl{\form}_{k+1}$}.
\noindent
Note that this set is finite: fixpoints are only expanded once.

\label{sec:lean}
The {\em Lean set of a formula} $\form$ includes navigating formulas of the form $\modalf{\modals}{\top}$, every navigating formulas of the form $\modalf{\modals}{\form^\prime}$ from the Fisher-Ladner closure, every proposition occurring in $\form$, written $\Props_{\form}$, every counting proposition, written $C$, and an extra proposition that does not occur in $\form$ used to represent other names, written $p_{\overline\form}$.

\begin{equation*}
\lean{\form}=\{\modalf{\modals}{\top}\}\cup \{\modalf{\modals}{\form^\prime}\in \flcl{\form}\} \cup \Props_{\form} \cup C \cup \{p_{\overline\form}\}
\end{equation*}

A {\em $\form\!-\!$node }, written $\fnode{\form}$, is a non-empty subset of $\lean{\form}$, such that:
\begin{itemize}
\item exactly one proposition from $\Props_{\form} \cup \{p_{\overline\form}\}$ is in each $\form\!-\!$node;
\item when $\modalf{\modals}{\form^\prime}\in\fnode{\form}$, then $\modalf{\modals}{\true}\in\fnode{\form}$; and
\item both $\modalf{\invfc}{\true}$ and $\modalf{\invns}{\true}$ cannot be in the same $\form\!-\!$node.
\end{itemize}
The set of $\form\!-\!$nodes is defined as $\fNodes{\form}$.

Intuitively, the formula corresponding to a node $\fnode{\form}$ is the following.
\[
\fnode{\form} = \bigwedge_{\psi \in \fnode{\form}} \psi \wedge \bigwedge_{\psi \in \lean{\form} \setminus \fnode{\form}} \neg \psi
\]

When the formula $\form$ under consideration is fixed, we often omit the superscript.

A \emph{\psitree{\form}} is either the empty tree $\emptyset$, or a triple $(\fnode{\form}, \stree_1,\stree_2)$ where $\stree_1$ and $\stree_2$ are \psitree{\form}s. 

\begin{figure}
  \begin{mathpar}
    \inferrule*{ }{\node \vdash^{\form} \top} \and
    \inferrule*{\restrictedform\in\node}{\node\vdash^{\form} \restrictedform} \and
    \inferrule*{\restrictedform \not\in \node}{\node\vdash^{\form} \neg \restrictedform} \and
    \inferrule*{\node\vdash^{\form} \restrictedform_1 \\ \node\vdash^{\form} \restrictedform_2}{\node\vdash^{\form} \restrictedform_1\wedge\restrictedform_2} \and
    \inferrule*{\node\vdash^{\form} \restrictedform_1}{\node\vdash^{\form} \restrictedform_1\vee\restrictedform_2} \and
    \inferrule*{\node\vdash^{\form} \restrictedform_2}{\node\vdash^{\form} \restrictedform_1\vee\restrictedform_2} \and
    \inferrule*{\node\vdash^{\form} \restrictedform\{\subst{\ufixpf{\var}{\restrictedform}}{\var}\}}{\node\vdash^{\form} \ufixpf{\var}{\restrictedform}}
   % \inferrule*{\node\nvdash^{\form}_{\stree,\pathn} \form^\prime}{\node\vdash^{\form}_{\stree,\pathn} \neg\form^\prime} \and
   % \inferrule*{\form^\prime \in \lean{\form} \\ \form^\prime\notin\node}{\node\nvdash^{\form}_{\stree,\pathn} \form^\prime} \and
%    \inferrule*{
    %\node\vdash^\form\nav{\trail}{\form}
%  }{\node\vdash^{\form}\countf{}{\trail}{\form}{\#}{\natn}} 
%    \inferrule*{\node\nvdash^{\form}_{\stree,\pathn} \form_1 \\ \node\nvdash^{\form}_{\stree,\pathn} \form_2}{\node\nvdash^{\form}_{\stree,\pathn} \form_1\vee\form_2} \and
%    \inferrule*{\node\nvdash^{\form}_{\stree,\pathn} \form_1}{\node\nvdash^{\form}_{\stree,\pathn} \form_1\wedge\form_2} \and
%    \inferrule*{\node\nvdash^{\form}_{\stree,\pathn} \form_2}{\node\nvdash^{\form}_{\stree,\pathn} \form_1\wedge\form_2} \and
%    \inferrule*{\node\nvdash^{\form}_{\stree,\pathn} \form^\prime\{\subst{\ufixpf{\var}{\form^\prime}}{\var}\}}{\node\nvdash^{\form}_{\stree,\pathn} \ufixpf{\var}{\form^\prime}}\and
%    \inferrule*{\node\vdash^{\form}_{\stree,\pathn} \form^\prime}{\node\nvdash^{\form}_{\stree,\pathn} \neg\form^\prime}
  \end{mathpar}
  \caption{Local entailment relation: between nodes and formulas}
  \label{fig:entailmentnode}
\end{figure}

We now turn to the definition of consistency of a \psitree{\form}. First, we define an entailment relation between a node and a formula in Figure~\ref{fig:entailmentnode}.
% The tree $\stree$ and path $\pathn$ annotations are used only when dealing with counting formulas, which are checked only when finished candidate \psitree{\form}s are generated (see below). We omit these annotations when they are not used.

We can now define the consistency relation between nodes of a \psitree{\form}.

  Two nodes $\node_1$ and $\node_2$ are consistent under modality $\modals \in \{\fc,\ns\}$, written $\fbrel{\form}{\node_1}{\modals}{\node_2}$, iff
  \begin{align*}
    \forall \modalf{\modals}\psi \in \lean{\form}&,\modalf{\modals}\psi \in \node_1 \iff \node_2 \vdash^{\form} \psi\\
    \forall \modalf{\dual\modals}\psi \in \lean{\form}&, \modalf{\dual\modals}\psi \in \node_2 \iff \node_1 \vdash^{\form} \psi
  \end{align*}

Consistency is checked each time a node is added to the tree, ensuring that forward modalities of the node are indeed satisfied by the nodes below, and that pending backward modalities of the node below are consistent with the added node. Note that do not check counting formulas at this point, as they are globally verified in the next step.

%In order to define a  consistent entailment relation for counting formulas in \psitree{\form}s, we need to consider several occurrences of the same nodes. 
%Hence, given a formula $\form$, we define the following multisets.
%\begin{align*}
%& \fNodes{\form}_0 &&=&& \{\node &&\mid&& \text{ $\node$ is a $\form\!-\!$node}\} \\
%& \fNodes{\form}_{i+1} &&=&& \fNodes{\form}_i \uplus \biguplus^{k}_{j=1}\{\node &&\mid&& \fishrel{\form^\prime}{\countf{}{\form_1}{\trail}{\form_2}{\#}{k}},\form^\prime\in\flcl{\form}_i, \\
%        &&&&& &&&&  \node \vdash^\form \nav{\trail}{\form_2}, \node\in \fNodes{\form}_i \}\\
%& \fNodes{\form} &&=&& \fNodes{\form}_k, &&&&\text{ where $k$ is the smallest integer s.t. }\flcl{\form}_k=\flcl{\form}_{k+1}
%\end{align*}

Upon generation of a finished tree, i.e., a tree with no pending backward modality, one may check whether a node of this tree satisfies $\form$. To this end, we first define forward navigation in a \psitree{\form} $\stree$. Given a path consisting of forward modalities $\pathn$, $\stree(\pathn)$ is the node at that path. It is undefined if there is no such node.
\begin{align*}
  (\node,\stree_1,\stree_2)(\epsilon) &= \node \\
  (\node,\stree_1,\stree_2)(\fc\pathn) &= \stree_1(\pathn) \\
  (\node,\stree_1,\stree_2)(\ns\pathn) &= \stree_2(\pathn)
\end{align*}
We also allow extending the path with backward modalities if they match the last modality of the path.
\begin{align*}
  (\node,\stree_1,\stree_2)(\pathn \fc \invfc) &= (\node,\stree_1,\stree_2)(\pathn) \\
  (\node,\stree_1,\stree_2)(\pathn \ns \invns) &= (\node,\stree_1,\stree_2)(\pathn)
\end{align*}

Now, we are able to define an entailment relation along paths in
\psitree{\form}s in Figure~\ref{fig:countingentailment}. This relation extends
local entailment relation (Figure~\ref{fig:entailmentnode}) with checks for
counting formulas. Note that the case for fixpoints is contained in the case for formulas with no counting subformula. Note also that $\neg \restrictedform$ in the ``less than'' case denotes the negation normal form.

\begin{figure}
  \begin{mathpar}
   % \inferrule*{\node^\prime \vdash^{\form}_{\stree,\pathn \modals} \form^\prime \\ \stree(\pathn \modals) = \node^\prime}{\node \vdash^{\form}_{\stree,\pathn} \modalf{\modals}\form^\prime} \and
   \inferrule*{\form^\prime \text{ does not contain counting formulas} \\ \stree(\pathn) \vdash^\form \form^\prime}{\pathn \vdash^{\form}_{\stree} \form^\prime}\and
    \inferrule*{\pathn\vdash^{\form}_{\stree}  \form_1 \\ \pathn\vdash^{\form}_{\stree}  \form_2}{\pathn\vdash^{\form}_{\stree}  \form_1\wedge\form_2} \and
    \inferrule*{\pathn\vdash^{\form}_{\stree}  \form_1}{\pathn\vdash^{\form}_{\stree}  \form_1\vee\form_2} \and
    \inferrule*{\pathn\vdash^{\form}_{\stree}  \form_2}{\pathn\vdash^{\form}_{\stree}  \form_1\vee\form_2} \and
    \inferrule*{\pathn \modals \vdash^{\form}_{\stree}\form^\prime}{\pathn \vdash^{\form}_{\stree} \modalf{\modals}\form^\prime}\and
    \inferrule*{
     | \{ \node^\prime,\; \pathn^\prime \in \trail \wedge \stree(\pathn \pathn^\prime) = \node^\prime \wedge \node^\prime \vdash^{\form} \restrictedform \land c \} | > \natn
}{ \pathn \vdash^{\form}_{\stree} \countf{c}{\trail}{\restrictedform}{>}{\natn}}
\and
    \inferrule*{
     | \{ \node^\prime,\; \pathn^\prime \in \trail \wedge \stree(\pathn \pathn^\prime) = \node^\prime \wedge \node^\prime \vdash^{\form} \restrictedform \land c \} | \leq \natn \\
     \forall \pathn^\prime \in \trail, 
     \stree(\pathn\pathn^\prime) \vdash^\form (\restrictedform \land c) \lor (\neg \restrictedform \land \neg c)
}{ \pathn \vdash^{\form}_{\stree} \countf{c}{\trail}{\restrictedform}{\leq}{\natn}}   \end{mathpar}
  \caption{Global entailment relation (incl. counting formulas)}
  \label{fig:countingentailment}
\end{figure}

We conclude these preliminaries by introducing some final notations.
The {\em root} of a \psitree{\form} is defined as follows. 
\begin{align*}
\stroot{\emptyset} & = \emptyset\\
\stroot{(\node,\stree_1,\stree_2)} & = \node
\end{align*}
We extend this notion to multiset of trees and write $\stroot{ST}$ for
the multiset of roots of the trees of $ST$.

The multiset of nodes of a tree is defined as follows.
\begin{align*}
\stnodes{\emptyset} &= \emptyset\\
\stnodes{(\node,\stree_1,\stree_2)} &= \{\node\} \cup \stnodes{\stree_1} \cup \stnodes{\stree_2}
\end{align*}
We also extend this notion to multiset of trees.

A \psitree{\form} $\stree$ {\em satisfies} a formula $\form$, written $\stree\vdash \form$, if neither $\modalf{\invfc}\top$ nor $\modalf{\invns}\top$ occur in $\stroot{\stree}$, and if there is a path $\pathn$ such that $\stree(\pathn) = \node$ and $\node\vdash^{\form}_{\stree,\pathn} \form$.

A multiset of trees $ST$ {\em satisfies} a formula $\form$, written $ST\vdash \form$, when there is a syntactic tree $\stree\in ST$ such that
$\stree\vdash \form$.

%%%%%%%%%%%%%%%%%%%%%%%%%%%%%%%%%%%%%%%%%%%%%%%%%%%%%%%%%%%%%%%%
\subsection{The Algorithm}
We are now ready to present the algorithm, which is parameterized by $K(\form)$, the maximum number of occurrences of a given node in a path from the root of the tree to a leaf. It builds consistent candidate trees from the bottom up, and checks at each step if one of the built tree satisfies the formula, returning $1$ if it is the case. As the set of nodes from which to build the trees is finite, it eventually stops and returns $0$ if no satisfying tree has been found.

\begin{algorithm}
\caption{Check Satisfiability of $\form$} 
\label{satalgo}
\begin{algorithmic}
    \STATE $ST \leftarrow \emptyset$
    \REPEAT
        \STATE $AUX \leftarrow \{(\node,\stree_1,\stree_2) \mid$ \COMMENT{we extend the trees}\\
        \quad $\nmax(\node,\stree_1,\stree_2) \leq K(\form) + 2$ \COMMENT{with an available node}\\
        \quad for $i$ in $\fc,\ns$ \COMMENT{and each child is either}\\
        \quad $\stree_i = \emptyset$ and $\modalf{i}{\top} \notin \node$ \COMMENT{an empty tree}\\
        \quad or $\stree_i\in ST$ \COMMENT{or a previously built tree}\\
        \quad\quad $\modalf{\dual{i}}\top \in \stroot{\stree_i}$ \COMMENT{with pending backward modalities}\\
        \quad\quad $\fbrel{\form}{\node}{i}{\stroot{\stree_i}}\}$ \COMMENT{checking consistency}
        \IF{$AUX \subseteq ST$}
      \RETURN{$0$} \COMMENT{No new tree was built}
        \ENDIF
      \STATE $ST \leftarrow ST \cup AUX$
    \UNTIL{$ST \vdash \form$}
    \RETURN{$1$}
\end{algorithmic}
\end{algorithm}

We now define the auxiliary $\nmax$ function as follows, where $\imax$ is the usual maximum function between integers.
\begin{align*}
  \nmax(\node,\stree_1,\stree_2) &= \imax(\nmax(\node,\stree_1),\nmax(\node,\stree_2))\\
  \nmax(\node,(\node,\stree_1,\stree_2)) & = 1 + \nmax(\stree_1,\stree_2)\\
  \nmax(\node,(\node^\prime,\stree_1,\stree_2)) & = \nmax(\stree_1,\stree_2) \quad\text{if $\node \neq \node^\prime$}\\
  \nmax(\node, \emptyset) &= 0
\end{align*}

Note a formula $\ufixpf{\var}{\form}$ can be rewritten in an equivalent formula such that $\var$ in $\form$ 
is only present in formulas with the form $\modalf{\modals}{\var}$. With this last observation, we now define the parameter for the number of occurrence of the same node in the tree in Figure \ref{fig:boundK}.

\begin{figure}
\begin{align*}
& K(\prop) = K(\neg\prop) = K(\neg \modalf{\modals}{\true}) = K(\true) = K(\var)=0 \\
& K(\form_1\wedge\form_2) = K(\form_1\vee\form_2)= K(\form_1)+K(\form_2) \\
& K(\modalf{\modals}{\form}) =  K(\ufixpf{\var}{\form}) = K(\form) \\
& K(\countf{}{\trail}{\restrictedform}{\#}{\natn}) = \natn+1
\end{align*}
\caption{Occurrences bound}
\label{fig:boundK}
\end{figure}

Consider for instance the formula $\form = \prop_1\wedge \modalf{\fc}{\countf{}{\ns^\star}{\prop_2}{>}{1}}$. The computed Lean is as follows, where $\psi = \ufixpf{\var}{\prop_2 \vee \modalf{\ns}{\var}} $.
\[\{\prop_1,\prop_2,\prop_3,\modalf{\fc}{\top},\modalf{\ns}{\top},\modalf{\invfc}{\top}, \modalf{\invns}{\top}, \modalf{\fc}{\psi}, \modalf{\ns}{\psi} \}\]

Proposition $\prop_3$ represents names other than $\prop_1$ and $\prop_2$.
We now compute the bound on nodes: $K = 2$.

After the first step, $ST$ consists of the trees of the form $(\{\prop_i\}, \emptyset, \emptyset)$ and $(\{\prop_i, \modalf{\dual{j}}{\top}\}, \emptyset, \emptyset)$, with $i \in \{1,2,3\}$ and $j \in \{\fc,\ns\}$. At this point the three finished trees in $ST$ are tested and found not to satisfy $\form$.

After the second iteration many trees are created, but the one of interest is the following.
\[\tree_0 = (\{\prop_2,\modalf{\ns}{\top},\modalf{\invfc}{\top}, \modalf{\ns}{\psi}\},\emptyset,(\{\prop_2,\modalf{\invns}{\top}\},\emptyset,\emptyset))\]

The third iteration yields the tree $(\{\prop_1, \modalf{\fc}{\psi}, \modalf{\fc}{\top}\},\tree_0,\emptyset)$, which is found to satisfy $\form$ at path $\epsilon$. As the nodes at every step are different, the limit is not reached. Figure \ref{algfig} depicts a graphical representation of the example
where counted nodes are drawn as thick circles.

\begin{figure}
\centering
\begin{tikzpicture}[scale=0.6]. 
\draw [thin,dotted] (-1,3) -- (9,3);
\draw [thin,dotted] (-1,5) -- (9,5);
\draw [thin,dotted] (-1,7) -- (9,7);
\node (p1) at (0,2) [circle,draw] {$\prop_1$};
\node (p2) at (2,2) [circle,draw] {$\prop_2$};
\node (p3) at (4,2) [circle,draw] {$\prop_3$};
%\node (p4) at (4,2)  {$\ldots$};
\node (p10) at (5,2)  {$\ldots$};
\node (p5) at (7,2) [circle,draw,very thick] {$\prop_2$};
\node (p6) at (6,2)  {$\ldots$};
\node (p7) at (6,4) [circle,draw,very thick] {$\prop_2$};
\node (p8) at (5,6) [circle,draw,very thick] {$\prop_2$};
\node (p9) at (6,8) [circle,draw] {$\prop_1$};
\draw [->] (p9) to node[auto] {$\fc$} (p8);
\draw [->] (p8) to node[auto] {$\ns$} (p7);
\draw [->] (p7) to node[auto] {$\ns$} (p5);
\end{tikzpicture}
\caption{Checking  $\form = \prop_1\wedge \modalf{\fc}{\countf{}{\ns^\star}{\prop_2}{>}{2}}$} \label{algfig}
\end{figure}

\subsection{Termination}
Proving termination of the algorithm is straightforward, as only a finite number of trees may be built and the algorithm stops as soon as it cannot build a new tree.

\subsection{Soundness}

If the algorithm terminates with a candidate, we show that the initial formula is satisfiable. Let $\stree,\pathn$ the \psitree{\form} and path such that $\pathn \vdash^\form_\stree \form$. We extract a tree from $\stree$ and show that the interpretation of $\form$ for this tree includes the node at path $\pathn$.

We write $\tree(\stree)$ for the tree $(\Nodes,\brel,\tlabel)$ defined as follows. We first rewrite $\Gamma$ such that each node $\node$ is replaced by the path to reach it.
\begin{align*}
  path(\node,\stree_1,\stree_2) & \rightarrow (\epsilon, path(\fc,\stree_1), path(\ns,\stree_2))\\
  path(\pathn, (\node,\stree_1,\stree_2)) & \rightarrow (\pathn, path(\pathn \fc,\stree_1), path(\pathn \ns,\stree_2))\\
  path(\pathn, \emptyset) & \rightarrow \emptyset
\end{align*}

We then define:
\begin{itemize}
\item $\Nodes = \stnodes{path(\stree)}$;
\item for every $(\pathn,\stree_1,\stree_2)$ in $path(\stree)$ and $i=\fc,\ns$,
 if $\stree_i\neq\emptyset$ then $R(\pathn,i)=\pathn i$ and $R(\pathn i, \dual{i}) = \pathn$; and
\item for all $\pathn\in\Nodes$ if $\prop\in\stree(\pathn)$ then $\tlabel(\pathn)=\prop$.
\end{itemize}

\begin{lemma}\label{lem:localsound}
  Let $\restrictedform$ a subformula of $\form$ with no counting formula. If $\stree(\pathn) \vdash^\form \restrictedform$ then we have $\pathn \in \semf{\restrictedform}{\tree(\stree)}{\emptyset}$.
\end{lemma}

\begin{proof}
  We proceed by induction on the lexical ordering of the number of unfolding of $\restrictedform$ that are required for $\tree(\stree)$, and of the size of the formula.
  
  The base cases are $\true$, atomic or counting propositions, and negated forms. These are immediate by definition of $\semf{\restrictedform}{\tree(\stree)}{\emptyset}$. The cases for disjunction and conjunction are immediate by induction (the formula is smaller). The case for fixpoints is also immediate by induction, as the number of unfoldings required decreases, and as $\semf{\ufixpf{\var}{\restrictedform}}{\tree(\stree)}{\emptyset} = \semf{\restrictedform\{\subst{\ufixpf{\var}{\restrictedform}}}{\tree(\stree)}{\emptyset}$.
  
  The last case is the presence of a modality $\modalf{\modals}\restrictedform$ from the $\form$node $\stree(\pathn)$. In this case we rely on the fact that the nodes $\stree(\pathn \modals)$ and $\stree(\pathn)$ are consistent to derive $\pathn \modals \vdash^\form \restrictedform$. We then conclude by induction.
\end{proof}

\begin{theorem}[Soundness]\label{sound} If
$\pathn \vdash^\form_\stree \form$ then $\pathn \in \semf{\form}{\tree(\stree)}{\emptyset}$
\end{theorem}

\begin{proof}
  We proceed by induction on the derivation of $\pathn \vdash^\form_\stree \form$.
  
  The proof is a consequence of the more general result $\pathn^\prime \vdash^\form_\stree \form^\prime \implies \pathn^\prime \in \semf{\form^\prime}{\tree(\stree)}{\emptyset}$ for any subformula of $\form^\prime$, by induction on the derivation of $\stree(\pathn^\prime) \vdash^\form_{\stree,\pathn} \form^\prime$. If $\form^\prime$ has no counting formula, the result is immediate by Lemma~\ref{lem:localsound}. Most cases are immediate by induction. As concerns the case for counting formulas, each hypothesis $\node^\prime \vdash^\form \restrictedform \land c$ has as hypothesis $\node^\prime \vdash^\form \restrictedform$. This is enough to conclude by induction for the ``greater than'' case. For the ``less than'' case, every node that is not counted has to satisfy $\neg \restrictedform \land \neg c$, so in particular $\neg \restrictedform$, and we conclude by induction.
\end{proof}

\subsection{Completeness}

Our proof proceeds in two step. We build a \psitree{\form} that satisfies the
formula, then we proceed to show it is actually built by the algorithm.

Assume that formula $\form$ is satisfiable by a tree
$\tree$. We consider the smallest such tree (i.e., the tree with the fewest
number of nodes) and fix $\node^\star$, a node witnessing satisfiability.

We now build a \psitree{\form} homomorphic to $\tree$, called the Lean labeled
version of $\form$, written $\stree(\tree,\form)$, and defined as follows.

First, we annotate counted nodes along with their corresponding counting proposition, yielding a new tree $\tree_c$. Starting from $\node^\star$ and by induction on $\form$, we proceed as follows. For formulas with no counting subformula, including recursion, we stop. For conjunction and disjunction of formulas, we recursively annotate according to both subformulas. For modalities, recursively annotate from the node under the modality. For $\countf{c}{\trail}{\restrictedform}{\leq}{\natn}$, we annotate every selected node with the counting proposition corresponding to the formula. For $\countf{c}{\trail}{\restrictedform}{>}{\natn}$, we annotate exactly $\natn+1$ selected nodes.

We now extend the semantics of formulas to take into account counting propositions and annotated nodes, written $\semfc{\cdot}{\tree}{\valuation}$. The definition is identical to Figure \ref{formsem}, with one addition and two changes. The addition is for counting propositions, which we define as $\node \in \semfc{c}{\tree}{\valuation}$ iff $\node$ is annotated by $c$. The two changes are for counting propositions, which we define as follows, selected only nodes that are annotated.

\begin{align*}
\semfc{\countf{}{\alpha}{\form^\prime}{\leq}{\natn}}{\tree}{\valuation} &=
    \{\node , |\{\node^\prime\in \semfc{\form^\prime}{\tree}{\valuation} \cap \semfc{c}{\tree}{\valuation} , \pathh{\node}{\alpha}{\node^\prime}\}|\leq  \natn \}\\
\semfc{\countf{}{\alpha}{\form^\prime}{>}{\natn}}{\tree}{\valuation} &=
    \{\node ,  |\{\node^\prime\in \semfc{\form^\prime}{\tree}{\valuation} \cap \semfc{c}{\tree}{\valuation} , \pathh{\node}{\alpha}{\node^\prime}\} |>  \natn \}
\end{align*}

We show that this modification of the semantics does no change the satisfiability of the formula.

\begin{lemma}\label{lem:annotated_satisfiability}
  We have $\node^\star \in \semfc{\form}{\tree}{\emptyset}$.
\end{lemma}

\begin{proof}
  We proceed by recursion on the derivation $\node^\star \in \semf{\form}{\tree}{\emptyset}$. The cases where no counting formula is involved, thus including fixpoints, are immediate, as the selected nodes are identical. The disjunction, conjunction, and modality cases are also immediate by induction. The interesting cases are the counting formulas.
  
  For $\countf{c}{\trail}{\restrictedform}{>}{\natn}$, as there are exactly $\natn+1$ nodes annotated, the property is true by induction. For $\countf{c}{\trail}{\restrictedform}{\leq}{\natn}$, we rely on the fact that every counted node is annotated. We conclude by remarking that $\restrictedform$ does not contain a counting formula, thus we have $\semfc{\restrictedform}{\tree}{\valuation} = \semf{\restrictedform}{\tree}{\valuation}$ and $\semfc{\neg \restrictedform}{\tree}{\valuation} = \semf{\neg \restrictedform}{\tree}{\valuation}$.
\end{proof}

To every node $\node$, we associate $\fnode{\form}$, a subset of formulas of the Lean selecting the node.
\begin{equation*}
\fnode{\form} = \{\form_0\mid \node \in \semf{\form_0}{\tree}{\emptyset} ,\form_0\in\lean{\form}\}
\end{equation*}

Note that this is a $\form$-node as it contains one and exactly one proposition, and if it includes a modal formula $\modalf{\modals}{\psi}$, then  it also includes $\modalf{\modals}{\top}$.

The tree $\stree(\tree,\form)$ is then built homomorphically to $\tree$.

In the remainder of this section, we write $\stree$ for $\stree(\tree,\form)$.
We first check that $\stree$ is consistent, starting with local consistency.

\begin{figure}
  \begin{mathpar}
    \inferrule*{\psi \in \lean{\form}}{\psi \induced \lean{\form}}\and
    \inferrule*{\psi_1 \induced \lean{\form} \\ \psi_2 \induced \lean{\form}}{\psi_1 \land \psi_2 \induced \lean{\form}} \and
    \inferrule*{\psi_1 \induced \lean{\form} \\ \psi_2 \induced \lean{\form}}{\psi_1 \lor \psi_2 \induced \lean{\form}} \and
    \inferrule*{ }{\true \induced \lean{\form}} \and
    \inferrule*{\psi \in (\Props_{\form} \cup \modalf{\modals}\true \cup C) }{\neg\psi \induced \lean{\form}}
  \end{mathpar}
  \caption{Formula induced by a lean}
  \label{fig:induced_by_lean}
\end{figure}

In the following, we say a formula $\psi$ is induced by the lean of $\form$,
written $\psi \induced \lean{\form}$, if it consists of the conjunction and
disjunction of formulas from the lean as defined in
Figure~\ref{fig:induced_by_lean}.

\begin{lemma}\label{lem:formulas_induced}
  Let $\modalf{\modals}\psi$ be a formula in $\lean{\form}$, and let $\psi^\prime$ be $\psi$ after unfolding its fixpoint formulas not under modalities. We have $\psi^\prime \induced \lean{\form}$.
\end{lemma}

\begin{proof}
  By definition of the lean and of the $\induced$ relation.
\end{proof}

\begin{lemma}\label{lem:complete_entailment}
  Let $\psi$ be a formula induced by $\lean{\form}$. We have $\node \in \semfc{\psi}{\tree}{\emptyset}$ if and only if $\fnode{\form} \vdash^{\form} \psi$.
\end{lemma}

\begin{proof}
   We proceed by induction on $\psi$. The base cases (the formula is in the $\form$-node or is a negation of a lean formula not in the $\form$-node) hold by definition of $\fnode{\form}$. The inductive cases are straightforward as these formulas only contain fixpoints under modalities.
\end{proof}

\begin{lemma}\label{lem:complete_consistent}
  Let $\node_1$ and $\node_2$ such that $\brel(\node_1,\modals) = \node_2$ with $\modals \in \{\fc,\ns\}$. We have $\fbrel{\form}{\fnode{\form}_1}{\modals}{\fnode{\form}_2}$.
\end{lemma}

\begin{proof}
  Let $\modalf{\modals}\psi$ be a formula in $\lean{\form}$. We show that $\modalf{\modals}\psi \in \fnode{\form}_1 \iff \fnode{\form}_2 \vdash^\form \psi$. We have $\modalf{\modals}\psi \in \fnode{\form}_1$ if and only if 
  $\node_1 \in \semfc{\modalf{\modals}\psi}{\tree}{\emptyset}$ by definition of $\fnode{\form}_1$, which in turn holds if and only if $\node_2 = \brel(\node_1,\modals) \in \semfc{\psi}{\tree}{\emptyset}$. 
We now consider $\psi^\prime$ which is $\psi$ after unfolding its fixpoint formulas not under modalities. We have $\semfc{\psi^\prime}{\tree}{\emptyset} = \semfc{\psi}{\tree}{\emptyset}$ and we conclude by Lemmas \ref{lem:formulas_induced} and \ref{lem:complete_entailment}.
\end{proof}

We now turn to global consistency, taking counting formulas into account.

\begin{lemma}\label{comple1}
  Let $\form_s$ be a subformula of $\form$, and $\pathn$ be a path from the root in $\tree$ such that $\tree(\pathn) \in \semfc{\form_s}{\tree}{\emptyset}$. We then have $\pathn \vdash^{\form}_{\stree} \form_s$.
\end{lemma}

\begin{proof}
  We proceed by induction on $\form_s$.
  
  If $\form_s$ does not contain any counting formula, we consider $\form_s^\prime$ which is $\form_s$ ofter unfolding its fixpoint formulas not under modalities. We have $\semfc{\form_s^\prime}{\tree}{\emptyset} = \semfc{\form_s}{\tree}{\emptyset}$ and $\form_s^\prime \induced \lean{\form}$. We conclude by Lemma \ref{lem:complete_entailment}.
  
  For most inductive cases, the proof is immediate by induction, as  the formula size decreases.
  
  For $\countf{c}{\trail}{\restrictedform}{>}{\natn}$, we have by induction form every counted node $\stree(\pathn\pathn^\prime) \vdash^\form \restrictedform$ and $\stree(\pathn\pathn^\prime) \vdash^\form c$. We conclude by the conjunction rule and by the counting rule of Figure \ref{fig:countingentailment}.
    
  For $\countf{c}{\trail}{\restrictedform}{\leq}{\natn}$, we proceed as above for the counted nodes. For the nodes that are not counted, have $\semfc{\neg\restrictedform}{\tree}{\valuation} = \semf{\neg\restrictedform}{\tree}{\valuation}$ and by soundness, we have $\stree(\pathn\pathn^\prime) \vdash^\form \neg\restrictedform$. We conclude by remarking that the node is not annotated by $c$, hence $\stree(\pathn\pathn^\prime) \vdash^\form \neg c$.
\end{proof}

We now need to show that the \psitree{\form} $\stree$ is actually built by the algorithm. The proof that it is the case follows closely the one from \cite{geneves-pldi07}, with a crucial exception: we need to make sure there are enough instances of each formula. Indeed, in \cite{geneves-pldi07}, the algorithm uses a $\form$type (a subset of $\lean{\form}$) at most once on each branch from the root to a leaf of the built tree. This yields a simple condition to stop the algorithm and conclude the formula is unsatisfiable. However, in the presence of counting formulas, a given $\form$type may occur more than once on a branch. To maintain the termination of the algorithm, we bound the number of identical $\form$type that may be needed by $K(\form)$ as defined in Figure \ref{fig:boundK}. We now check that this bound is sufficient to build a tree for any satisfiable formula.

We recall that $\form$ is a satisfiable formula and $\tree$ is a smallest tree such that $\form$ is satisfied, and $\node^\star$ is a witness of satisfiability.

We proceed in two steps: first we show that counted nodes (with counted propositions) imply a bound on the number of identical $\form$types on a branch for a smallest tree. Second, we show that this minimal marking is bound by $K(\form)$.

In the following, we call counted nodes and node $\node^\star$ \emph{annotations}.

We now define the projection of an annotation on a path. Let $\pathn$ be a path from the root of the tree to a leaf. An annotation projects on $\pathn$ at $\pathn_1$ if $\pathn = \pathn_1 \pathn_2$, the annotation is at $\pathn_1 \pathn_m$, and $\pathn_2$ shares no prefix with $\pathn_m$.

\begin{lemma}\label{lem:annotation_present}
  Let $\stree^\prime$ be the annotated tree, $\pathn$ a path from the root of the tree to a leaf, $\node_1$ and $\node_2$ two distinct nodes of $\pathn$ such that $\fnode{\form}_1 = \fnode{\form}_2$. Then either annotations projects both on $\pathn$ at $\node_1$ and $\node_2$, or an annotation projects strictly between $\node_1$ and $\node_2$.
\end{lemma}

\begin{proof}
  We proceed by contradiction: we assume there is no annotation that projects between $\node_1$ and $\node_2$ and at most one of them has an annotation that projects on it. Without loss of generality, we assume that $\node_2$ is below $\node_1$ in the tree.
  
  Assume neither $\node_1$ nor $\node_2$ is annotated (through projection). We show that the tree where $\brel(\node_1,\fc) \leftarrow \brel(n_2,\fc)$ and $\brel(\node_1,\ns) \leftarrow \brel(\node_2,\ns)$ still satisfies $\form$ at $\node$, a contradiction since this tree is strictly smaller. Let $\tree_s$ be this smaller tree, $\stree_s$ the corresponding \psitree{\form}, and for every path $\pathn$ of $\stree$, let $\pathn_s$ be the potentially shorter path if it exists (i.e., if it was not removed when pruning the tree). More precisely, let $\pathn_1$ be the path to $\node_1$ and $\pathn_1\pathn_2$ be the path to $\node_2$. If $\pathn^\prime = \pathn^\prime_1 \pathn^\prime_3$ where $\pathn^\prime_1$ is a prefix of $\pathn_1$ and the paths are disjoint from there, then $\stree_s(\pathn^\prime) = \stree(\pathn^\prime)$. If $\pathn^\prime = \pathn_1 \pathn_2 \pathn_3$, then $\stree_s(\pathn_1\pathn_3) = \stree(\pathn^\prime)$.
  
  First, as there was no annotation projected, $\node$ is still part of this tree at a path $\pathn_s$. We show that we have $\pathn_s \vdash^{\form}_{\stree_s} \form$ by induction on the derivation $\pathn \vdash^{\form}_{\stree} \form$. Let $\pathn^\prime \vdash^{\form}_{\stree} \form^\prime$ in the derivation, assuming that $\pathn^\prime_s$ is defined.
  
  The case where $\form^\prime$ does not mention any counting formula is trivial: $\stree(\pathn^\prime) = \stree_s(\pathn^\prime_s)$ thus local entailment is immediate.
  
  Conjunction and disjunction are also immediate by induction.
  
  For the modality case, we first need to prove an additional property. If $\pathn^\prime \vdash^{\form}_{\stree} \modalf{\modals} \form^\prime$ and $\form^\prime$ contains a counting formula, then $\pathn^\prime \modals$ is either a prefix of $\pathn_1$ followed by a disjoint path, or it includes $\pathn_1\pathn_2$. We prove this property by contradiction. The formula $\modalf{\modals} \form^\prime$ is both in $\stree(\pathn_1)$ and in $\stree(\pathn_1\pathn_2)$. We consider the outermost counting formula in $\form^\prime$ which we write $\form^\prime_c$. It presence implies the occurrence of a counting proposition $c$ in the formula. Since counting propositions are distinct for distinct syntactic occurrences of a formula, this implies that the corresponding counting proposition is either under a fixpoint (which is impossible), or under an enclosing counting formula, which is also impossible. We thus have a contradiction
  
  % We are thus in a ``counting under counting'' case, and $\form^\prime_c$ occurs as a subformula $\form''_c$ that is being counted. As $\form''_c$ is mentioned both in  $\stree(\pathn_1)$ and in $\stree(\pathn_1\pathn_2)$ (as a subformula of $\psi$), this implies that either the external counting formula is itself under counting, or that a fixpoint is present because of a navigation involving a repeated trail. Since the number of counting formulas is finite, the latter case has to be true for the outermost counting formula, thus contradicting the fact that formula under consideration is a modality. Thus $\pathn^\prime \modals$ is still a valid path for $\stree_s$, and we conclude by induction.
  
We now turn to the counting case $\countf{c}{\trail}{\restrictedform}{\#}{\natn}$. We say that a path \emph{does not cross over} when this path does not contain $\node_1$ nor $\node_2$. For nodes that are reached using paths that do not cross over, we conclude by induction that they are also counted. We now show that the remaining nodes for which a crossover happened are also reached. Without loss of generality, assume that $\pathn^\prime$ is a prefix of $\pathn_1$ (the counting formula is in the ``top'' part of the tree), and let $\pathn_n$ be the path from the counting formula to the counted node ($\pathn_n$ is an instance of the trail $\trail$). This path is of the shape $\pathn^\prime_1 \pathn_2 \pathn_c$, with $\pathn_1 = \pathn^\prime \pathn^\prime_1$. We now show that the path $\pathn^\prime_1 \pathn_c$ is an instance of $\trail$ if and only if $\pathn_n$ is, thus the same node is still counted.

Recall that $\trail$ is of the shape $\trail_1, \ldots, \trail_n, \trail_{n+1}$ where $\trail_1$ to $\trail_n$ are of the form $\trail_r^\star$ and where $\trail_{n+1}$ does not contain a repeated trail. We say that a prefix $\pathn_p$ of a path $\pathn$ \emph{stops at $i$} if there is a suffix $\pathn_s$ such that $\pathn_p\pathn_s$ is still a prefix of $\pathn$, if $\pathn_p\pathn_s \in \trail_1, \ldots, \trail_i$, and if there is no shorter suffix $\pathn^\prime_s$ and $j$ such that $\pathn_p\pathn^\prime_s \in \trail_1, \ldots, \trail_j$. (Intuitively, $\trail_i$ is the trail being used when matching the end of $\pathn_p$.) Note that $i$ may not be unique as a path may be matched in different ways by a trail. We now show that there are $i \leq j \leq n$ such that both $\pathn^\prime_1$ stops at $i$ and $\pathn^\prime_1 \pathn_2$ stop at $j$. We thus show that $j$ cannot be $n+1$. Recall that $\trail_{n+1}$ does not contain a repeated subtrail. If $\fnode{\form}_2$ does not contain the counted proposition $c$ (which may happen in the case of a ``less than'' counting where the target is not counted), then neither does $\fnode{\form}_1$, which is a contradiction to the fact that $\trail{i},\ldots,\trail_{n+1}$ is not empty (in that case the counted proposition is necessarily mentioned). Thus $\fnode{\form}_2$ contain formulas without a fixpoint (as the trail is not repeated) mentioning $c$. Consider the largest such formula. By an induction on the path $\pathn_2$, we build a strictly larger formula that occurs in $\fnode{\form}_1$. This a contradiction to the hypothesis that $\fnode{\form}_1 = \fnode{\form}_2$.

We now consider the suffixes $\pathn_s^1$ and $\pathn_s^2$ computed when stating that the paths stop at $i$ and $j$. These suffixes correspond to the path matching the end of $\trail_i$ and $\trail_j$, respectively (before the next iteration or switching to the next formula). They have matching formulas in $\fnode{\form}_1$ and $\fnode{\form}_2$. As the formulas are present in both nodes, then the remainder of the paths ($\pathn_2\pathn_c$ and $\pathn_c$) are instances of $(\pathn_s^1 | \pathn_s^2) \trail_i \ldots \trail_{n+1}$, thus $\pathn^\prime_1 \pathn_c$ is an instance of $\trail$ if and only if $\pathn_n$ is.

In the case of ``greater than'' counting, we conclude immediately by induction as the same nodes are selected (thus there are enough). In the case of ``less than'', we need to check that no new node is counted in the smaller tree. Assume it is not the case for the formula $\countf{}{\trail}{\restrictedform}{\leq}{\natn}$, thus there is a path $\pathn_n \in \trail$ to a node satisfying $\restrictedform$. As the same node can be reached in $\stree$, and as we have $\stree(\pathn^\prime \pathn_n) \vdash^{\form} \neg \restrictedform$ by induction, we have a contradiction.

This concludes the proof when neither $\node_1$ nor $\node_2$ is annotated. The proof is identical when $\node_2$ is annotated. If $\node_1$ is annotated, we look at the first modality between $\node_1$ and $\node_2$. If it is a $\fc$, then we build the smaller tree by doing $\brel(\node_1,\fc) \leftarrow \brel(n_2,\fc)$ (we remove the $\ns$ subtree from $\node_2$ instead of $\node_1$). Symmetrically, if the first modality is a $\ns$, we consider  $\brel(\node_1,\ns) \leftarrow \brel(\node_2,\ns)$ as smaller tree. The rest of the proof proceeds as above.
\end{proof}

\begin{theorem}[Completeness]\label{thm:completeness}
  If $\form$ is satisfiable, then a satisfying tree is built.
\end{theorem}

\begin{proof}
  The proof proceeds as in \cite{geneves-pldi07}, we only need to check there are enough copies of each node to build every path. Let $\pathn$ be a path from the root of the tree to the leaves. By Lemma~\ref{lem:annotation_present}, there are at most $n+1$ identical nodes in this path, where $n$ is the number of marks. The number of marks is $c+1$ where $c$ is the number of counted nodes. We show by an immediate induction on the formula $\form$ that $c$ is bound by $K(\form)$ as defined in Figure~\ref{fig:boundK}. We conclude by remarking that $K(\form) + 2$ is the number of identical nodes we allow in the algorithm.
\end{proof}

\subsection{Complexity} \label{complexity}

We now show that the complexity of the satisfiability algorithm is exponential time w.r.t. the formula size. This is achieved in two steps: we first show that the Lean size is linear w.r.t. the formula size,
then we show that the algorithm has a single exponential complexity w.r.t. to the Lean size.

\begin{lemma}
The Lean size is linear in terms  of the original formula size.
\end{lemma}
\begin{proof}[Proof Sketch]
It was shown in \cite{geneves-pldi07} that the Lean size of non counting formulas is linear with respect to the formula size.

We now describe the case for counting formulas.
Note that each counting formula introduces only one new counting proposition in the Lean. 
A first duplication of formulas is considered in the construction of the Lean for  "less than" counting formulas. Both, the formula witnessing the counted nodes and its negation are considered. 
Furthermore, another duplication  is introduced for counting formulas of the form $\countf{}{\trail_1|\trail_2}{\form}{\#}{k}$.
Each of these duplications only doubles the size of the Lean. Hence, the Lean size remains linear w.r.t to the original formula size.
%There is a duplication of formulas occurring after trails of the form $\trail_1\mid\trail_2$, hence, the Lean size is now, considering counting formulas, at most the double of the formula size. 
%Furthermore, in "less than" counting formulas,  both, the formula witnessing the nodes aimed to be counted, and its negation are considered in the construction of the Lean, 
%there is an additional subformula per modality of the initial formula (at most doubling the lean size, again).
\end{proof}

\begin{theorem}\label{complexitythm}

The satisfiability algorithm for the logic is decidable in time $2^{O(n)}$, where $n$ is the Lean size. % and $m$ is the greatest nesting level of counting formulas.

\end{theorem}
\begin{proof}[Proof Sketch]
The cardinality of nodes set  is $2^n$. The number of occurrences of each node in the tree is bounded by $K(\form)\leq k*m$, where
$k$ is the greatest constant occurring in the counting formulas and $m$ is the number of counting subformulas. Hence the number of steps in the algorithm is bounded by $2^n*k*m$.

As for the functions at each step, $\nmax$ is a single traversal to the tree.
Since the entailment relation involved in the definition of $R^\form$ is only local, $R^\form$ is performed in linear time.

The number of choices to form trees (triples) at each step is restricted by $3*(2^n*k*m)$.

The global entailment relation involves four exponential time traversals: the number of trees, the number of nodes at each tree,
the number of traversals for the entailment relation of counting formulas, and the cost of each of such traversals. 
Hence it takes no more than $4*(2^n*k)$ time.

%The complexity for non counting formulas, that is, where there a occurrences of formulas with the form $\countf{}{\trail}{\form}{\#}{k}$, is single exponential with respect to the Lean size \cite{geneves-pldi07}.
%the level of nesting of counting formulas is a new source of complexity in the decision procedure,but it still remains exponential w.r.t. the formula size.

\end{proof}

Theorem~\ref{complexitythm} states the complexity for the logic defined in Figure~\ref{normalformsyn}. We now state that the same complexity upper-bound  holds if we additionally consider counting formulas of the form
$\modalf{\fc}{\countf{}{\ns^*}{\form}{\#}{k}}$ in the scope of a fixpoint operator (as presented in Section~\ref{sec:gradedpaths}).

\begin{lemma} \label{countinfixp}
Given a formula $\form$ where counting subformulas $\psi$ only count children nodes,
if every counting subformula $\psi$ is replaced by the equivalent fixpoint formula $\uc(\psi)$ in $\form$, $\form[\subst{\uc(\psi)}{\psi}]$,
then $Lean(\form[\subst{\uc(\psi)}{\psi}])\leq Lean(\form)*k^l$, where $k$ is greatest numerical constraint of the counting subformulas,
and $l$ is the greatest level nesting of counting subformulas.
\end{lemma}
\begin{proof}[Proof Sketch]
It is proven by induction on the structure of $\form$, and in the case of counting formulas,  another induction is done on the numerical constraint.
\end{proof}

\begin{corollary}
The logic supporting counting formulas only on children in the scope of fixpoint formulas or another counting formula is decidable in $2^{O(n*k^l)}$, where
$k$ is the greatest cardinality constraint and $l$ is the greatest nesting level of counting formulas.
\end{corollary}
\begin{proof}
Immediate from Theorem~\ref{complexitythm} and Lemma~\ref{countinfixp}.
\end{proof}

%\begin{corollary}
%The logic supporting counting formulas only on children in the scope of fixpoint formulas is decidable in $2^{O(n*k^m+m)}$, where
%$k$ is the greatest cardinality constraint and $m$ is the greatest nesting level of counting formulas.
%\end{corollary}
%\begin{proof}
%Immediate from Theorem~\ref{complexitythm} and Lemma~\ref{countinfixp}.
%\end{proof}

%%%%%%%%%%%%%%%%%%%%%%%%%%%%%%%%%%%%%%%%%%%%%%%%%%%%%%%%%%%%%%%%%%%%%%%%%

\section{Application to XML Trees} \label{sec:application} %Regular Paths and Counting

\subsection{XPath Expressions}\label{typespaths} \label{sec:xpath}

XPath \cite{xpath} was introduced as part of the W3C XSLT transformation language to have a non-XML format for selecting nodes and computing values from an XML
document (see \cite{geneves-pldi07} for a formal presentation of XPath). Since then XPath has become part of
several other standards, in particular it forms the ``navigation subset'' of the XQuery language.

In their simplest form XPath expressions look like ``directory navigation paths''.  For example, the XPath
\begin{verbatim}
  /company/personnel/employee
\end{verbatim}
navigates from the root of a document through the top-level ``company'' node to its
``personnel'' child nodes and on to its ``employee'' child nodes.  The result of
the evaluation of the entire expression is the set of all the ``employee''
nodes that can be reached in this manner.  At each step in the navigation, the selected nodes for that step can be
filtered with a test. Of special interest to us are the predicates that test node's count or the selected node's position in the previous step's selection.  For example, if we ask for
\begin{verbatim}
  /company/personnel/employee[position()=2]
\end{verbatim}
then the result is \emph{all} employee nodes that are the \emph{second} employee node among
the employee child nodes of each personnel node selected by the previous step.

XPath also makes it possible to combine the capability of searching along
``axes'' other than the shown ``children of'' with counting constraints. 
%The essential XPath axes are illustrated on~Figure~\ref{fig:xpath-axes}.
% \begin{figure}[h]
%\begin{center}
%\includegraphics[keepaspectratio=true, width=8cm]{simple-partition-no-shadow.pdf}
%\end{center}
%\caption{XPath axes: partition of tree nodes.}\label{fig:xpath-axes}
%\end{figure}
For example, if we ask for
\begin{verbatim}
/company[count(descendant::employee<=300)]/name
\end{verbatim}
then the result consists of the company names with less than 300 employees in total (the axis ``descendant'' is the transitive closure of the default -- and often omitted -- axis ``child''). 
%This illustrates the expressive power of the XPath language when extended with counting features.

The syntax and semantics of Core XPath expressions are respectively given on Figure~\ref{fig:xpath-syntax} and Figure~\ref{fig:xpath-semantics}.
An XPath expression is interpreted as a relation between nodes. The considered XPath fragment allows absolute and relative paths, path union, intersection, composition, as well as node tests and qualifiers with counting operators, conjunction, disjunction, negation, and path navigation. Furthermore, it supports all XPath axes allowing multidirectional navigation.

\begin{figure}[h]
\begin{align*}
   \Axis      \syntaxdefinition &  \text{self} \mid \text{child} \mid \text{parent} \mid \text{descendant} \mid \text{ancestor} \mid\\ 
   &   \text{following-sibling}    \mid \text{preceding-sibling} \mid \\ 
   & \text{following} \mid \text{preceding} \\
  % ----
  \NameTest  \syntaxdefinition &  \QName \mid * \\
   % -----
  \Step  \syntaxdefinition & \Axis\text{::}\NameTest \\
   % -----
   \PathExpr  \syntaxdefinition &  \PathExpr/\PathExpr
\mid \PathExpr[\Qualifier] \mid \Step \\
 \Qualifier  \syntaxdefinition & \PathExpr \mid \CountExpr \mid\qualifnot \Qualifier \mid \\
  & \Qualifier \qualifand \Qualifier  \mid \Qualifier \qualifor \Qualifier \mid \nominal \\
  \CountExpr \syntaxdefinition & \qualifcount{\PathExpr'}~\Comparison~k \\ 
% no counting under counting:
 \PathExpr'  \syntaxdefinition &  \PathExpr'/\PathExpr'
\mid \PathExpr'[\Qualifier'] \mid \Step \\
\Qualifier'  \syntaxdefinition & \PathExpr' \mid\qualifnot \Qualifier' \mid \Qualifier' \qualifand \Qualifier' \\ 
& \mid \Qualifier' \qualifor \Qualifier' \mid \nominal \\
 \Comparison  \syntaxdefinition & \leq  \mid > \mid \geq   \mid < \mid =\\
  \XPath \syntaxdefinition & \PathExpr \mid /\PathExpr  \mid \XPath \pathunion \PathExpr \mid\\
   & \XPath \pathintersect \PathExpr  \mid \XPath \pathexcept \PathExpr 
   \end{align*}
   \caption{Syntax of Core XPath Expressions.}\label{fig:xpath-syntax}
\end{figure}

\begin{figure}
\begin{align*}
   \pathsem{ \Axis\text{::}\NameTest }      =& \{ (x,y) \in \dom^2 \mid x(\Axis)y \text{ and }  \\ 
                                                            & y \text{ satisfies } \NameTest  \} \\
   \pathsem{ /\PathExpr }    =&  \{ (r,y) \in \pathsem{\PathExpr } \mid \\ & r \text{ is the root}   \} \\
  % ----
   \pathsem{ P_1/P_2}    =& \pathsem{P_1} \circ \pathsem{P_2}\\
    % ----
   \pathsem{P_1 \pathunion P_2}    =&  \pathsem{P_1} \cup \pathsem{P_2}\\
     % ----
   \pathsem{ P_1 \pathintersect P_2 }    =& \pathsem{P_1} \cap \pathsem{P_2}\\
     % ----
   \pathsem{ P_1 \pathexcept P_2 }    =& \pathsem{P_1} \setminus \pathsem{P_2}\\
     % ----
   \pathsem{ \PathExpr[\Qualifier] }    = & \{ (x,y) \in \pathsem{\PathExpr} \mid \\ & y \in \qualifsem{\Qualifier} \}\\ \\
     % ----
   \qualifsem{  \PathExpr  }    =& \{ x \mid  \exists y. (x,y) \in \pathsem{\PathExpr}  \} \\
      %---
   \qualifsem{\qualifcount{\PathExpr}~\Comparison~k } =& \{ x \in \dom \mid \\& \card{\left\{y \mid (x,y) \in \pathsem{\PathExpr} \right\}} \\ & \text{satisfies }  \Comparison~k \} \\
     % ----
   \qualifsem{  \qualifnot Q  }   =& \dom \setminus \qualifsem{Q}\\
     % ----
   \qualifsem{Q_1 \qualifand Q_2}    =& \qualifsem{Q_1} \cap \qualifsem{Q_1} \\
     % ----
   \qualifsem{ Q_1 \qualifor Q_2 }   =&  \qualifsem{Q_2} \cup \qualifsem{Q_2}
   \end{align*}
   \caption{Semantics of Core XPath Expressions}\label{fig:xpath-semantics}
 \end{figure}

It was already observed in \cite{DBLP:conf/doceng/GenevesR05,Balder09} that using positional information in paths reduces to counting (at the cost of an exponential blow-up).
For example, the expression
\begin{verbatim}
child::a[position()=5]
\end{verbatim}
first selects the ``\texttt{a}'' nodes occurring as children of the current context node, and then keeps those occurring at the $5$th position. This expression can be rewritten into the semantically equivalent expression:
\begin{verbatim}
child::a[count(preceding-sibling::a)=4]
\end{verbatim}
which constraints the number of preceding siblings named ``\texttt{a}'' to $4$, so that the  qualifier becomes true only for the $5$th child ``\texttt{a}''. A general translation of positional information in terms of counting operators \cite{DBLP:conf/doceng/GenevesR05,Balder09} is summarized on Figure~\ref{fig:positional-info}, where $\precedence$ denotes the document order (depth-first left-to-right) relation in a tree.  Note that translated path expressions can in turn be expressed into the core XPath fragment of Figure~\ref{fig:xpath-syntax} (at the cost of another exponential blow-up). Indeed, expressions like $\PathExpr/(\PathExpr_2 \pathexcept \PathExpr_3)/\PathExpr_4$ must be rewritten into expressions where binary connectives for paths occur only at top level, as in:
\begin{align*}
& \PathExpr/\PathExpr_2/\PathExpr_4 \pathexcept \\ & \PathExpr/\PathExpr_3/\PathExpr_4
\end{align*}

\begin{figure}[h]
\begin{align*}
\PathExpr[\qualifposition=1] \equiv &  \PathExpr \pathexcept (\PathExpr/\precedence) \\
\PathExpr[\qualifposition=k+1] \equiv &  (\PathExpr \pathintersect \\ & (\PathExpr[k]/\!\precedence))[\qualifposition\!=\!1] \\
\precedence \equiv&  (\text{descendant::*})  \pathunion (\text{a-o-s::*}/ \\ &  \text{following-sibling::*}/\text{d-or-s::*}) \\
\text{a-or-s::*}   \equiv & \text{ancestor::*} \pathunion \text{self::*}\\
\text{d-or-s::*} \equiv & \text{descendant::*} \pathunion \text{self::*}
\end{align*}
   \caption{Positional Information as Syntactic Sugars  \cite{DBLP:conf/doceng/GenevesR05,Balder09}}\label{fig:positional-info}
\end{figure}

We focus on Core XPath expressions involving the counting operator  (see Figure~\ref{fig:xpath-syntax}). The XPath fragment without the counting operator (the navigational fragment) was already linearly translated into $\mu$-calculus in  \cite{geneves-pldi07}. The contributions presented in this paper allow to equip this navigational fragment with counting features such as the ones formulated above. Logical formulas capture the aforementioned XPath counting constraints.
For example, consider the following XPath expression:
\begin{verbatim}
child::a[count(descendant::b[parent::c])>5]
\end{verbatim}
This expression selects the children nodes named ``\texttt{a}'' provided they have more than $5$
descendants which (1) are named ``\texttt{b}'' and (2) whose parent is named ``\texttt{c}''. The logical formula denoting the set of children nodes
named ``\texttt{a}'' is:
\[\psi = a\wedge \modalf{\invns^*,\invfc}\top  \]
The  logical translation of the above XPath expression is:
%\[\psi\wedge\countf{}{\fc,(\fc|\ns)^\star}{(b \wedge \mu y.\modalf{\invfc}c \vee \modalf{\invns}y)}{>}{5}\]
\[\psi\wedge\modalf{\fc}{\countf{}{(\fc|\ns)^\star}{(b \wedge \ufixpf{\var}{\modalf{\invfc}{c}\vee \modalf{\invns}{\var}} )}{>}{5}}\]
%\modalf{\invns^*,\invfc}c)}{>}{5}}
This formula holds for nodes selected by the XPath expression. A correspondence between the main XPath axes over unranked trees and modal formulas over binary trees is given in Figure~\ref{fig:axes-modalities}. In this figure, each logical formula holds for nodes selected by the corresponding XPath axis from a context $\gamma$.
\begin{figure}[h]
  \begin{center}
$\begin{array}{r|l}
  \text{Path} & \text{Logical formula}\\
  \hline
 \gamma/\text{self::*}  &  \gamma  \\
 \gamma/\text{child::*} &  \modalf{\invns^*, \invfc}\gamma  \\
 \gamma/\text{parent::*} &  \modalf{\fc}{\modalf{\ns^*}{\gamma}} \\
 \gamma/\text{descendant::*} &  \modalf{(\invns\mid\invfc)^*, \invfc}\gamma\\
 \gamma/\text{ancestor::*}  &  \modalf{\fc}{\modalf{(\fc\mid\ns)^*}{\gamma}}\\
 \gamma/\text{following-sibling::*} &  \modalf{\invns}{\modalf{\invns^*}{\gamma}}\\
 \gamma/\text{preceding-sibling::*} & \modalf{\ns}{\modalf{\ns^*}{\gamma}} \\
% \psi/\text{following::*} & \\
% \psi/\text{preceding::*} & 
\end{array}$
\end{center}
\caption{XPath axes as modalities over binary trees.}\label{fig:axes-modalities}
\end{figure}

%The general translation of a counting XPath expression is given on Figure~\ref{fig:exprs}. 

%\begin{lemma}
%Core XPath expressions with counting constraints can be linearly translated into the logic.
%\label{xpathtranslationlemma} 
%\end{lemma}

Let consider another example (XPath expression $e_1$):
\begin{verbatim}
child::a/child::b[count(child::e/descendant::h])>3]
\end{verbatim}
Starting from a given context in a tree, this XPath expression navigates to children nodes named ``a'' and selects their children named ``b''. Finally, it retains only those ``b'' nodes for which the qualifier between brackets holds.
The first path can be translated in the logic as follows:
%$$ \vartheta = b \wedge \modalf{\invns^*, \invfc} (a \wedge \modalf{\invns^*, \invfc}\top) $$
$$ \vartheta = b \wedge \ufixpf{\var}{\modalf{\invfc}{(a \wedge \ufixpf{\var^\prime}{\modalf{\invfc}{\top} \vee \modalf{\invns}{\var^\prime}}}) \vee \modalf{\invfc}{\var}}$$

This example requires a more sophisticated translation in the logic. This is because it makes implicit that ``e'' nodes (whose existence is simply tested for counting purposes) must be children of selected ``b'' nodes. The translation of the full aforementioned XPath expression is as follows:
$$\vartheta \wedge \nominal \wedge \countf{}{(\invfc\mid\invns)^*, (\fc\mid\ns)^*}{\eta}{>}{3}$$  
where $\nominal$ is a new fresh nominal used to mark a ``b'' node which is filtered by the qualifier and the formula $\eta$ describes the counted ``h'' nodes:
%$$\eta = h \wedge \modalf{(\invns\mid\invfc)^*, \invfc}(e \wedge \modalf{\invns^*, \invfc}\nominal)$$
$$\eta = h \wedge \ufixpf{\var}{\modalf{\invfc}{(e \wedge \ufixpf{\var^\prime}{\modalf{\invfc}{\nominal} \vee \modalf{\invns}{\var^\prime}}}) \vee \modalf{\invns}{\var} \vee \modalf{\invfc}{\var}}$$
Intuitively, the general idea behind the translation is to first translate the leading path, use a fresh nominal for marking a node which is filtered, then find at least ``3'' instances of ``h'' nodes from which we can reach back the marked node via the inverse path of the counting formula. 

Since trails make it possible to navigate but not to test properties (like existence of labels), we test for labels in the counted formula $\eta$ and we use a general navigation $(\invfc\mid\invns)^*, (\fc\mid\ns)^*$ to look for counted nodes everywhere in the tree. Introducing the nominal is necessary to bind the context properly (without loss of information). Indeed, the XPath expression $e_1$ makes implicit that a ``e'' node must be a child of a ``b'' node selected by the outer path. Using a nominal, we restore this property by connecting the counted nodes to the initial single context node.

\begin{lemma}
The translation of Core XPath expressions with counting constraints into the logic is linear.
\label{xpathtranslationlemma}
\end{lemma}
It is proven by structural induction in a similar manner to \cite{geneves-pldi07} (in which the translation is proven for expressions without counting constraints). For counting formulas, the use of nominals and the general (constant-size) counting trail make it possible to avoid duplication of trails so that the translation remains linear.

\begin{corollary}
The equivalence problem for expressions of the form:
\[\PathExpr[\qualifcount{\PathExpr'} \# \natn]   \]
\noindent
where $\# \in \{\leq, >, =\}$ and $\natn$ is a constant, is decidable. More specifically, the equivalence problem can be decided in exponential time in terms of the expression size and the highest nesting level of counting formulas.
\end{corollary}

\subsection{Regular Tree Languages with Cardinality Constraints}

Regular tree grammars capture most of the schemas in use today \cite{DBLP:journals/toit/MurataLMK05}. The logic can express all regular tree languages (it is easy to prove that regular expression types in the manner of e.g., \cite{DBLP:journals/toplas/HosoyaVP05}  can be linearly translated into the logic: see \cite{geneves-pldi07}). 

In practice, schema languages often provide shorthands for expressing cardinality constraints on node occurrences. XML Schema notably offers two attributes {\em minOccurs} and {\em maxOccurs} for this purpose. For instance, the following XML schema definition:
\begin{Verbatim}[fontsize=\small]
<xsd:element name="a">
  <xsd:complexType>
    <xsd:sequence>
     <xsd:element name="b" minOccurs="4" maxOccurs="9"/>
    </xsd:sequence>
  </xsd:complexType>
</xsd:element>
\end{Verbatim} 
\noindent
is a notation that restricts the number of occurrences of ``\texttt{b}'' nodes to be at least 4 and at most 9, as children of ``\texttt{a}'' nodes. The goal here is to have a succinct notation for expressing regular languages which could otherwise be exponentially large if written with usual regular expression operators. The above regular requirement can be translated as the formula:
\[\phi \wedge \modalf{\fc}{(\countf{}{\ns^\star}{b}{>}{3}} \wedge \countf{}{\ns^\star}{b}{\leq}{9})\]
where $\phi$ corresponds to the regular tree type $a[b^*]$ as follows:
\[\begin{array}{ll}
\phi =&(a \wedge (\neg\modalf{\fc}{\top} \vee \modalf{\fc}{\psi})) \wedge \neg\modalf{\ns}{\top}  \vspace{0.1cm}\\
\psi =&\mu x. \left(b \wedge \neg\modalf{\fc}{\top} \wedge \neg\modalf{\ns}{\top}\right)  \vee \left(b \wedge \neg\modalf{\fc}{\top} \wedge \modalf{\ns}{x}\right)
\end{array}\]

%The logic also supports occurrence constraints on regular expressions (permitted by XML Schema \cite{schema-part1} too).
% For example, consider the regular expression $e=a^{*}ba^{*}$. To require that the regular language denoted by $e$ additionally contains at least two elements $a$, one must express it 
%as $e'=aaa^{*}ba^{*} | aa^{*}baa^{*} | a^{*}baaa^{*}$.  Alternatively, one may write 
% \[\phi_e\wedge\countf{}{\ns^\star}{a}{>}{1} \] where $\phi_e$ is the logical translation of $e$.

This example only involves counting over children nodes. The logic allows counting through more general trails, and in particular arbitrarily deep trails. 
Trails corresponding to the XPath axes ``preceding, ancestor, following'' can be used to constrain the context of a schema. % (see~Figure~\ref{fig:xpath-axes}).
The ``descendant'' trail can be used to specify additional constraints over the subtree defined by a given schema. For instance, suppose we want to forbid webpages containing nested anchors ``$a$'' (whose interpretation makes no sense for web browsers). We can build the logical formula $f$ which is the conjunction of a considered schema for webpages (e.g. XHTML) with the formula $a/\text{descendant::}a$ in XPath notation. Nested anchors are forbidden by the considered schema iff $f$ is unsatisfiable.

As another example, suppose we want paragraph nodes (``$p$'' nodes) not to be nested inside more than 3 unordered lists (``$ul$'' nodes), regardless of the schema defining the context. One may check for the unsatisfiability of the following formula:
\[ p\wedge \countf{}{(\invfc|\invns)^\star,\invfc}{ul}{>}{3} \]

\section{Related Work}  \label{sec:relatedwork}

\paragraph{Counting over graphs}
The $\mu$-calculus is a propositional modal logic augmented with least and greatest fixpoint operators \cite{DBLP:conf/icalp/Kozen82}. Kupferman, Sattler and Vardi study a $\mu$-calculus with graded modalities where one can 
express, e.g., that a node has at least $n$ successors satisfying a certain property \cite{DBLP:conf/cade/KupfermanSV02}. The 
modalities are limited in scope since they only count children of a given node. 

The $\mu$-calculus has been recently extended with inverse modalities \cite{685998}, nominals \cite{DBLP:conf/cade/SattlerV01}, and graded modalities \cite{DBLP:conf/cade/KupfermanSV02}. If only two of the above constructs are considered,  satisfiability of the enriched calculus is EXPTIME-complete \cite{Bonatti-et-al-ICALP-06}.
However, if all of the above constructs are considered simultaneously, the calculus becomes undecidable \cite{Bonatti-et-al-ICALP-06}. Hopefully, this undecidability result for the case of graphs does not preclude decidable tree logics combining such features.

\paragraph{Counting over trees}
%Automata capable of expressing structural (schema) constraints and queries on finite trees were reported in  \cite{DBLP:conf/dbpl/CalvaneseGLV09}. Even more, it is able to perform backwards and recursive navigation, as well as to interpret nominals.

The notion of Presburger Automata for trees, combining both regular 
constraints on the children of nodes and numerical constraints given by Presburger 
formulas, has independently been introduced by Dal Zilio and Lugiez \cite{964013} and Seidl et al. \cite{DBLP:conf/icalp/SeidlSMH04}.  Specifically, Dal Zilio and Lugiez \cite{964013} propose a modal logic for unordered trees called Sheaves logic. This logic allows to impose certain arithmetical constraints on children nodes but lacks recursion (i.e., fixpoint operators) and inverse navigation.  Dal Zilio and Lugiez consider the satisfiability and the membership problems. 
%and they show  that Sheaves logic formulas can be translated into deterministic automata and is decidable.
 Demri and Lugiez \cite{DBLP:conf/cade/DemriL06} showed by means of an automata-free decision procedure that this logic is only  PSPACE-complete. Restrictions like {\em $\prop_1$ nodes have no more ``children'' than $\prop_2$ nodes}, are expressible by this approach. Seidl et al. \cite{DBLP:conf/icalp/SeidlSMH04} introduce a fixpoint Presburger logic, which, in addition to numerical constraints on  children nodes, also supports recursive forward navigation. For example, expressions like {\em the descendants of $\prop_1$ nodes have no more ``children'' than the number of children of descendants of $\prop_2$ nodes} are allowed.
This means that constraints can be imposed on sibling nodes (even if they are deep in the tree) by forward recursive navigation but not on distant nodes which are not siblings. 

Compared to the work presented here, neither of the two previous approaches can support constraints like {\em there are more than $5$ ancestors of ``$\prop$'' nodes}. 

Furthermore, due to the lack of backward navigation, the works found in \cite{964013,DBLP:conf/icalp/SeidlSMH04,DBLP:conf/cade/DemriL06} are not suited for succinctly capturing  XPath expressions. Indeed, it is well-known that expressions with backward modalities are exponentially more succinct than their forward-only counterparts \cite{symmetry,DBLP:conf/doceng/GenevesR05}.

%For instance the path $\text{child::}a/\text{child::}b$ selects the nodes labeled by $b$ with an $a$ node as parent, which also has a parent. If we qualify such path $\text{self::}a/\text{child::}b[\text{child::}c/\text{child::}d]$, we filter such $b$ nodes to also have $c$ children with at least a $d$ child.
%Then, we focus on $b$ nodes, and there is a need to express a backwards navigation to express the path composition properties, as for the qualifying properties,
%they are expressed by downwards navigation.
%Although two-way navigation usually does not increase expressivity, it may come at an exponential cost \cite{685998}.

There is poor hope to push the decidability envelope much further for counting constraints.
Indeed, it is known from \cite{DBLP:conf/icalp/KlaedtkeR03,DBLP:conf/cade/DemriL06,Balder09} that the equivalence problem is undecidable  for XPath expressions with counting operators of the form:
\begin{itemize}
\item $\PathExpr_1[\qualifcount{\PathExpr_2} = \qualifcount{\PathExpr_3}]$,  or
\item $\PathExpr_1[\qualifposition = \qualifcount{\PathExpr_2}]$.
\end{itemize}
This is the reason why logical frameworks that allow comparisons between counting operators limit counting by restricting the $\PathExpr$  to immediate children nodes \cite{964013,DBLP:conf/icalp/SeidlSMH04}.
In this paper, we chose a different tradeoff: comparisons are restricted to constants but at the same time comparisons along more general paths are permitted. 

\section{Conclusion} \label{sec:conclusion}
We introduced a modal logic of trees equipped with (1) converse modalities, which allow to succinctly express forward and backward navigation, (2) a least fixpoint operator for recursion, and (3) cardinality constraint operators for expressing numerical occurrence constraints on tree nodes satisfying some regular properties. A sound and complete algorithm is presented for testing satisfiability of logical formulas.
This result is surprising since the corresponding logic for graphs is undecidable \cite{Bonatti-et-al-ICALP-06}. 

The decision procedure for the logic is exponential time w.r.t. to the formula size.
The logic captures regular tree languages with cardinality restrictions, as well as the navigational fragment of XPath equipped with counting features.  Similarly to backward modalities, numerical constraints do not extend the logical expressivity beyond regular tree languages. Nevertheless they enhance the succinctness of the formalism as they provide useful shorthands for otherwise exponentially large formulas.

This makes it possible to extend static analysis to a larger set of XPath and XML schema features in a more efficient way. We believe the field of application of this logic may go beyond the XML setting. For example, in verification of linked data structures \cite{kuncak08,DBLP:conf/tacas/HabermehlIV06} reasoning on tree structures with in-depth cardinality constraints seems a major issue. Our result may help building solvers that are attractive alternatives to those based on non-elementary logics such as SkS \cite{Thatcher68}, like, e.g., Mona \cite{mona_user_manual}.

%\paragraph{Future Work}
%We are currently implementing the satisfiability algorithm as an extension of our earlier solver introduced in \cite{geneves-pldi07,implementation}. 
%TODO: injective trails
%
%We say a trail is injective if its interpretation is an injection in any tree, that is, 
%if $\node\neq\node^\prime$, then %$\semf{\trail}{\tree}{\node}\cap\semf{\trail}{\tree}{\node^\prime}=\emptyset$.
%For example, the trail denoting the child/parent relation ($\fc,\ns^\star$) is injective: parents do not share children.

\bibliographystyle{alpha}
\bibliography{ctl}

\newcommand{\etalchar}[1]{$^{#1}$}
\begin{thebibliography}{MLMK05}

\bibitem[BLMV06]{Bonatti-et-al-ICALP-06}
Piero Bonatti, Carsten Lutz, Aniello Murano, and Moshe Vardi.
\newblock The complexity of enriched $\mu$-calculi.
\newblock In Michele Bugliesi, Bart Preneel, Vladimiro Sassone, and Ingo
  Wegener, editors, {\em ICALP}, volume 4052 of {\em Lecture Notes in Computer
  Science}, pages 540--551. Springer-Verlag, 2006.

\bibitem[CD99]{xpath}
J.~Clark and S.~DeRose.
\newblock {XML} path language ({XPath}) version 1.0.
\newblock {W3C} recommendation, November 1999.
\newblock {http://www.w3.org/TR/1999/REC-xpath-19991116}.

\bibitem[CGS09]{ghelli-icdt09}
Dario Colazzo, Giorgio Ghelli, and Carlo Sartiani.
\newblock Efficient asymmetric inclusion between regular expression types.
\newblock In {\em ICDT '09: Proceedings of the 12th International Conference on
  Database Theory}, pages 174--182, New York, NY, USA, 2009. ACM.

\bibitem[DL06]{DBLP:conf/cade/DemriL06}
S.~Demri and D.~Lugiez.
\newblock Presburger modal logic is {PSPACE-C}omplete.
\newblock In U.~Furbach and N.~Shankar, editors, {\em IJCAR}, volume 4130 of
  {\em Lecture Notes in Computer Science}, pages 541--556. Springer, 2006.

\bibitem[DZLM04]{964013}
Silvano Dal-Zilio, Denis Lugiez, and Charles Meyssonnier.
\newblock A logic you can count on.
\newblock In Neil~D. Jones and Xavier Leroy, editors, {\em POPL}, pages
  135--146. ACM, 2004.

\bibitem[Gel08]{DBLP:conf/mfcs/Gelade08}
Wouter Gelade.
\newblock Succinctness of regular expressions with interleaving, intersection
  and counting.
\newblock In Edward Ochmanski and Jerzy Tyszkiewicz, editors, {\em MFCS},
  volume 5162 of {\em Lecture Notes in Computer Science}, pages 363--374.
  Springer, 2008.

\bibitem[GGM09]{DBLP:conf/mfcs/GeladeGM09}
Wouter Gelade, Marc Gyssens, and Wim Martens.
\newblock Regular expressions with counting: Weak versus strong determinism.
\newblock In Rastislav Kr{\'a}lovic and Damian Niwinski, editors, {\em MFCS},
  volume 5734 of {\em Lecture Notes in Computer Science}, pages 369--381.
  Springer, 2009.

\bibitem[GLS07]{geneves-pldi07}
P.~Genev\`es, N.~Laya\"ida, and A.~Schmitt.
\newblock Efficient static analysis of {XML} paths and types.
\newblock In {\em PLDI}, pages 342--351, New York, NY, USA, 2007. ACM Press.

\bibitem[GMN08]{Gelade-siam08}
Wouter Gelade, Wim Martens, and Frank Neven.
\newblock Optimizing schema languages for xml: Numerical constraints and
  interleaving.
\newblock {\em SIAM J. Comput.}, 38(5):2021--2043, 2008.

\bibitem[GR05]{DBLP:conf/doceng/GenevesR05}
Pierre Genev{\`e}s and Kristoffer~H{\o}gsbro Rose.
\newblock Compiling {XP}ath for streaming access policy.
\newblock In Anthony Wiley and Peter~R. King, editors, {\em DocEng}, pages
  52--54. ACM, 2005.

\bibitem[HIV06]{DBLP:conf/tacas/HabermehlIV06}
Peter Habermehl, Radu Iosif, and Tom{\'a}s Vojnar.
\newblock Automata-based verification of programs with tree updates.
\newblock In Holger Hermanns and Jens Palsberg, editors, {\em TACAS}, volume
  3920 of {\em Lecture Notes in Computer Science}. Springer, 2006.

\bibitem[HJJ{\etalchar{+}}95]{klarlund-tacas95}
J.G. Henriksen, J.~Jensen, M.~J{\o}rgensen, N.~Klarlund, B.~Paige, T.~Rauhe,
  and A.~Sandholm.
\newblock Mona: Monadic second-order logic in practice.
\newblock In {\em Tools and Algorithms for the Construction and Analysis of
  Systems, First International Workshop, TACAS '95, LNCS 1019}, 1995.

\bibitem[HVP05]{DBLP:journals/toplas/HosoyaVP05}
H.~Hosoya, J.~Vouillon, and B.~C. Pierce.
\newblock Regular expression types for {XML}.
\newblock {\em ACM Trans. Program. Lang. Syst.}, 27(1):46--90, 2005.

\bibitem[KM01]{mona_user_manual}
Nils Klarlund and Anders M{\o}ller.
\newblock {\em {MONA Version 1.4 User Manual}}.
\newblock BRICS Notes Series NS-01-1, Department of Computer Science,
  University of Aarhus, January 2001.

\bibitem[Koz82]{DBLP:conf/icalp/Kozen82}
D.~Kozen.
\newblock Results on the propositional $\mu$-{C}alculus.
\newblock In M.~Nielsen and E.~M. Schmidt, editors, {\em ICALP}, volume 140 of
  {\em Lecture Notes in Computer Science}, pages 348--359. Springer, 1982.

\bibitem[KR03]{DBLP:conf/icalp/KlaedtkeR03}
F.~Klaedtke and H.~Rue{\ss}.
\newblock Monadic second-order logics with cardinalities.
\newblock In J.~C.~M. Baeten, J.~K. Lenstra, J.~Parrow, and G.~J. Woeginger,
  editors, {\em ICALP}, volume 2719 of {\em Lecture Notes in Computer Science},
  pages 681--696. Springer, 2003.

\bibitem[KSV02]{DBLP:conf/cade/KupfermanSV02}
O.~Kupferman, U.~Sattler, and M.~Y. Vardi.
\newblock The complexity of the graded $\mu$-calculus.
\newblock In A.~Voronkov, editor, {\em CADE}, volume 2392 of {\em Lecture Notes
  in Computer Science}, pages 423--437. Springer, 2002.

\bibitem[KT07]{Kilpelainen-ic07}
Pekka Kilpeläinen and Rauno Tuhkanen.
\newblock One-unambiguity of regular expressions with numeric occurrence
  indicators.
\newblock {\em Information and Computation}, 205(6):890 -- 916, 2007.

\bibitem[Mar05]{marx-tods05}
Maarten Marx.
\newblock Conditional xpath.
\newblock {\em ACM Trans. Database Syst.}, 30(4):929--959, 2005.

\bibitem[MLMK05]{DBLP:journals/toit/MurataLMK05}
M.~Murata, D.~Lee, M.~Mani, and K.~Kawaguchi.
\newblock Taxonomy of {XML} schema languages using formal language theory.
\newblock {\em ACM Trans. Internet Techn.}, 5(4):660--704, 2005.

\bibitem[MS72]{DBLP:conf/focs/MeyerS72}
Albert~R. Meyer and Larry~J. Stockmeyer.
\newblock The equivalence problem for regular expressions with squaring
  requires exponential space.
\newblock In {\em FOCS}, pages 125--129. IEEE, 1972.

\bibitem[OMFB02]{symmetry}
Dan Olteanu, Holger Meuss, Tim Furche, and Francois Bry.
\newblock {XPath}: Looking forward.
\newblock In {\em EDBT '02: Proceedings of the Worshop on {XML}-Based Data
  Management}, volume 2490 of {\em LNCS}, pages 109--127. Springer-Verlag,
  2002.

\bibitem[SSMH04]{DBLP:conf/icalp/SeidlSMH04}
H.~Seidl, T.~Schwentick, A.~Muscholl, and P.~Habermehl.
\newblock Counting in trees for free.
\newblock In J.~D\'{\i}az, J.~Karhum{\"a}ki, A.~Lepist{\"o}, and D.~Sannella,
  editors, {\em ICALP}, volume 3142 of {\em Lecture Notes in Computer Science},
  pages 1136--1149. Springer, 2004.

\bibitem[SV01]{DBLP:conf/cade/SattlerV01}
Ulrike Sattler and Moshe~Y. Vardi.
\newblock The hybrid $\mu$-calculus.
\newblock In {\em IJCAR}, pages 76--91, 2001.

\bibitem[tCM09]{Balder09}
Balder ten Cate and Maarten Marx.
\newblock Axiomatizing the logical core of {XPath} 2.0.
\newblock {\em Theor. Comp. Sys.}, 44(4):561--589, 2009.

\bibitem[TW68]{Thatcher68}
James~W. Thatcher and Jesse~B. Wright.
\newblock Generalized finite automata theory with an application to a decision
  problem of second-order logic.
\newblock {\em Mathematical Systems Theory}, 2(1):57--81, 1968.

\bibitem[Var98]{685998}
M.~Y. Vardi.
\newblock Reasoning about the past with two-way automata.
\newblock In {\em ICALP}, pages 628--641, London, UK, 1998. Springer-Verlag.

\bibitem[ZKR08]{kuncak08}
Karen Zee, Viktor Kuncak, and Martin Rinard.
\newblock Full functional verification of linked data structures.
\newblock In {\em PLDI}, pages 349--361, New York, NY, USA, 2008. ACM.

\end{thebibliography}

\end{document}